\documentclass[11pt]{amsart}
\usepackage[margin=1in]{geometry}
\usepackage[foot]{amsaddr}

\usepackage{amssymb,amsmath,amsthm,amsfonts,amscd,bbm}
\usepackage{enumitem}
\usepackage{pdflscape}
\usepackage{caption}
\usepackage[all]{xy}
\usepackage{hyperref}
\usepackage{xcolor}
\usepackage{graphicx}
\usepackage[outdir=./]{epstopdf}
\usepackage{enumitem}
\usepackage{tensor}
\usepackage{tikz-cd}
\usepackage{multicol}
\usepackage{mathtools}
\usepackage{tocvsec2}
\usepackage{bbm}
\usepackage{longtable}
\usepackage{fullpage}

\newcommand{\googlebooks}[1]{(preview at \href{https://books.google.com/books?id=#1}{google books})}

\newcommand{\numdam}[1]{}

\usepackage{mathrsfs}
\DeclareMathAlphabet{\mathpzc}{OT1}{pzc}{m}{it}

\def\semicolon{;}
\def\applytolist#1{
    \expandafter\def\csname multi#1\endcsname##1{
        \def\multiack{##1}\ifx\multiack\semicolon
            \def\next{\relax}
        \else
            \csname #1\endcsname{##1}
            \def\next{\csname multi#1\endcsname}
        \fi
        \next}
    \csname multi#1\endcsname}

\def\calc#1{\expandafter\def\csname c#1\endcsname{{\mathcal #1}}}
\applytolist{calc}QWERTYUIOPLKJHGFDSAZXCVBNM;
\def\bbc#1{\expandafter\def\csname bb#1\endcsname{{\mathbb #1}}}
\applytolist{bbc}QWERTYUIOPLKJHGFDSAZXCVBNM;
\def\bfc#1{\expandafter\def\csname bf#1\endcsname{{\mathbf #1}}}
\applytolist{bfc}QWERTYUIOPLKJHGFDSAZXCVBNM;
\def\sfc#1{\expandafter\def\csname s#1\endcsname{{\sf #1}}}
\applytolist{sfc}QWERTYUIOPLKJHGFDSAZXCVBNM;
\def\fc#1{\expandafter\def\csname f#1\endcsname{{\mathfrak #1}}}
\applytolist{fc}QWERTYUIOPLKJHGFDSAZXCVBNM;

\usepackage{tikz}
\usepackage{tikz-cd}

\makeatletter
\def\fixtikzforbreqn#1#2{%
  \protected\edef#1{\noexpand\ifmmode\mathchar\the\mathcode`#2 \noexpand\else#2\noexpand\fi}%
}
\fixtikzforbreqn\tikz@nonactivesemicolon;
\fixtikzforbreqn\tikz@nonactivecolon:
\fixtikzforbreqn\tikz@nonactivebar|
\fixtikzforbreqn\tikz@nonactiveexlmark!
\makeatother

\usepackage{breqn}   

\usetikzlibrary{arrows,backgrounds,patterns.meta}
\usetikzlibrary{positioning,shadings,cd}
\usetikzlibrary{shapes}
\usetikzlibrary{backgrounds}
\usetikzlibrary{decorations,decorations.pathreplacing,decorations.markings,decorations.pathmorphing}
\usetikzlibrary{fit,calc,through}
\usetikzlibrary{external}
\usetikzlibrary{arrows}
\tikzset{vertex/.style = {shape=circle,draw,fill=black,inner sep=0pt,minimum size=5pt}}
\tikzset{edge/.style = {->,> = latex', bend right}}
\tikzset{
	super thick/.style={line width=3pt}
}
\tikzset{
    quadruple/.style args={[#1] in [#2] in [#3] in [#4]}{
        #1,preaction={preaction={preaction={draw,#4},draw,#3}, draw,#2}
    }
}
\tikzstyle{shaded}=[fill=red!10!blue!20!gray!30!white]
\tikzstyle{unshaded}=[fill=white]
\tikzstyle{empty box}=[circle, draw, thick, fill=white, opaque, inner sep=2mm]
\tikzstyle{annular}=[scale=.7, inner sep=1mm, baseline]
\tikzstyle{rectangular}=[scale=.75, inner sep=1mm, baseline=-.1cm]
\tikzstyle{mid>}=[decoration={markings, mark=at position 0.5 with {\arrow{>}}}, postaction={decorate}]
\tikzstyle{mid<}=[decoration={markings, mark=at position 0.5 with {\arrow{<}}}, postaction={decorate}]
\tikzstyle{over}=[double, draw=white, super thick, double=]
\tikzstyle{snake}=[decorate, decoration={snake, segment length=1mm, amplitude=.3mm}]
\tikzstyle{saw}=[decorate, decoration={saw, segment length=.7mm, amplitude=.25mm}]

\tikzstyle{coupon}=[draw, very thick, rectangle, rounded corners=5pt]
\tikzset{Rightarrow/.style={double equal sign distance,>={Implies},->},
triplecd/.style={-,preaction={draw,Rightarrow}},
quadruplecd/.style={preaction={draw,Rightarrow,
shorten >=0pt
},
shorten >=1pt,
-,double,double
distance=0.2pt}}
\tikzset{
    tripleline/.style args={[#1] in [#2] in [#3]}{
        #1,preaction={preaction={draw,#3},draw,#2}
    }
}
\tikzstyle{triple}=[tripleline={[line width=.15mm,black] in
      [line width=.7mm,white] in
      [line width=1mm,black]}] 
\tikzset{
    quadrupleline/.style args={[#1] in [#2] in [#3] in [#4]}{
        #1,preaction={preaction={preaction={draw,#4},draw,#3}, draw,#2}
    }
}
\tikzstyle{quadruple}=[quadrupleline={[line width=.3mm,white] in
      [line width=.6mm,black] in
      [line width=1.2mm,white] in
      [line width=1.5mm,black]}]

\theoremstyle{plain}
\newtheorem{thm}{Theorem}[section]
\newtheorem*{thm*}{Theorem}

\newtheorem{cor}[thm]{Corollary}

\newtheorem*{cor*}{Corollary}

\newtheorem*{conj*}{Conjecture}

\newtheorem*{lem*}{Lemma}

\newtheorem{prop}[thm]{Proposition}
\newtheorem{quest}[thm]{Question}
\newtheorem*{quest*}{Question}
\newtheorem*{claim*}{Claim}

\theoremstyle{definition}
\newtheorem{defn}[thm]{Definition}

\newtheorem{ex}[thm]{Example}
\newtheorem{sub-ex}[thm]{Sub-Example}
\newtheorem{counter-ex}[thm]{Counter-Example}
\newtheorem{rem}[thm]{Remark}
\newtheorem*{rem*}{Remark}

\usepackage{xcolor}
\definecolor{dark-red}{rgb}{0.7,0.25,0.25}
\definecolor{dark-blue}{rgb}{0.15,0.15,0.55}
\definecolor{medium-blue}{rgb}{0,0,.8}
\definecolor{DarkGreen}{RGB}{0,150,0}
\definecolor{rho}{named}{red}
\hypersetup{
   colorlinks, linkcolor={purple},
   citecolor={medium-blue}, urlcolor={medium-blue}
}

\def\altdb{\vadjust{\vbox to 0pt{\vss\hbox{\kern \hsize
\quad{\dbend}}\kern\baselineskip\kern-10pt}}}

\newcommand{\noshow}[1]{}

\begin{document} 
\title{Universal coarse geometry of spin systems}
\author{Ali Elokl$^{1}$}
\address{$^{1}$ Department of Physics, North Carolina State University, Raleigh, NC 27695, USA}
\author{Corey Jones$^{2}$}
\address{$^{2}$ Department of Mathematics, North Carolina State University, Raleigh, NC 27695, USA}

\begin{abstract}
The prospect of realizing highly entangled states on quantum processors with fundamentally different hardware geometries raises the question: to what extent does a state of a quantum spin system have an intrinsic geometry? In this paper, we propose that both states and dynamics of a spin system have a canonically associated \textit{coarse geometry}, in the sense of Roe, on the set of sites in the thermodynamic limit. For a state $\phi$ on an (abstract) spin system with an infinite collection of sites $X$, we define a universal coarse structure $\mathcal{E}_{\phi}$ on the set $X$ with the property that a state has decay of correlations with respect to a coarse structure $\mathcal{E}$ on $X$ if and only if $\mathcal{E}_{\phi}\subseteq \mathcal{E}$. We show that under mild assumptions, the coarsely connected completion  $(\mathcal{E}_{\phi})_{con}$ is stable under quasi-local perturbations of the state $\phi$. We also develop in parallel a dynamical coarse structure for arbitrary quantum channels, and prove a similar stability result. We show that several order parameters of a state only depend on the coarse structure of an underlying spatial metric, and we establish a basic compatibility between the dynamical coarse structure associated to a quantum circuit $\alpha$ and the coarse structure of the state $\psi\circ \alpha$ where $\psi$ is any product state.
\end{abstract}

\maketitle

\tableofcontents

\section{Introduction}

Recently, there have been remarkable experimental realizations of highly entangled, topologically ordered ground states of gapped quantum spin liquids on quantum processors (\cite{doi:10.1126/science.abi8378, NonAbTop, NonAbTop2, NonAbTop3}). The states constructed in these experiments have (at least in principle) the \textit{same} entanglement structure as the ground state of the theoretical spin liquid up to local perturbation. Thus, properties of the spin liquid which emerge from the entanglement structure of the ground state alone, and are robust against local perturbation, are arguably \textit{realized} in the processor rather than merely simulated. A key example of such a property is topological order \cite{KITAEV20062, MR3929747}, 

Many of the interesting properties that we think of as emerging from the entanglement structure of a state also depend on the geometry of the lattice sites to define mathematically and in practice. In most of the experiments realizing topological order that have been performed so far, the spatial geometry of the processor's hardware roughly agrees with the geometry of the system being realized. However, for many states of interest, the geometry underlying the processor's hardware is neccessarily incompatible with the geometry of the system from which the state originates (which occurs to some extent in the experiment \cite{NonAbTop}).

To illustrate the problem, suppose that we prepare the ground state of a 3D spin liquid (e.g. a Walker-Wang model \cite{WW}) on a quantum processor with 2D hardware, making use of some (necessarily non-local) bijective encoding scheme to identify sites of the 3D lattice with sites in the 2D processor\footnote{in principle we could imagine a more complicated encoding scheme which doesn't directly map sites to sites, but we will see that allowing this can dramatically change the universal features of a system}. If we, as the experimenter, remember this encoding scheme, then it seems that the geometric incompatibility between the two systems may not be an issue, since the universal properties depending on the 2D metric of the system are still mathematically well defined as long as we utilize our encoded 3D metric. However, if we are simply handed the state on the processor and not told the encoding scheme, is it possible in principle to work out the ``correct" geometry on the set of sites to be used when defining and measuring characteristics such as topological order, independently of the physical spatial relations of the local Hilbert spaces?

\begin{quest}\label{question1}
Is there a precise mathematical sense in which the ``correct" geometry on the set of sites of a spin system, with respect to which entanglement-based order parameters should be defined, is intrinsic to a state? Or do we need to rely on the extrinsic spatial geometry of the sites to define entanglement-based order parameters?
\end{quest}

Another theoretical motivation for the above question comes from making sense of various notions of equivalence of quantum phases. Traditionally, we say ground states of two gapped Hamiltonians are in the same quantum phase if there is a (smooth, continuous, etc.) path of Hamiltonians between the two that keeps the gap open. However, the quantum information perspective has motivated the consideration of state-based definitions of equivalence that do not require reference to a parent Hamiltonian. Examples include finite-depth circuit equivalence, locality-preserving equivalence \cite{PhysRevB.82.155138, MR3929747}, and quasi-local automorphic equivalence \cite{AutomorphicEquiv}. We stress that all of these notions depend not just on the state, but a-priori also on the underlying geometry of the lattice to define. Thus the``intrinsic geometry" of a state should tell you which geometric structure you should be using when trying to characterize quantum phases.

In this paper, we propose a positive resolution to Question \ref{question1} in the thermodynamic limit using the mathematical framework of \textit{coarse geometry}.

\subsection{The coarse perspective.} 

An important way that the geometry of space is expressed in physical systems is through locality. This is typically encoded in the structure of the parent Hamiltonian of the system, by asking for interactions that appropriately decay as a function of distance \cite{MR1441540, MR3929747}. However, locality manifests in states (e.g. ground states or equilibrium/KMS states  \cite[Section 5.2]{MR1441540}) through the decay of the two point correlation functions as a function of distance. Here we do not refer to the actual \textit{rate} of decay of correlations, but to the mere \textit{fact} of decay as a function of distance. This leads us to say that a state of a spin system with a given metric on the set of sites is \textit{local} with respect to a metric function on the sites if the two-point correlation functions decay as a function of the distance between the supports of localized observables. 

The key observation of our paper is that this notion of locality for a state depends only on the abstract \textit{large scale geometry} of the metric, and is independent of the details at small length scales. Thus two metrics induce the same notion of locality if they ``look the same" at large length scales. This leads us to propose an answer to Question \ref{question1} using the framework of \textit{coarse geometry}, which formalizes the notion of large scale geometric structure. One way to make formal sense of two metrics looking the same at large length scales is via an equivalence relation on metric spaces called \textit{coarse equivalence}. For example, the metric space $\mathbbm{R}^{n}$ is coarsely equivalent to any discrete n-dimensional lattice (e.g. $\mathbbm{Z}^{n}$) despite obvious differences in topology, but $\mathbbm{R}^{n}$ is coarsely equivalent to $\mathbbm{R}^{m}$ if and only if $n=m$ (see Section \ref{coarsestructures} for a detailed explanation). Furthermore hyperbolic $n$-space $\mathbbm{H}^{n}$ is \textit{not} coarsely equivalent to $\mathbbm{R}^{n}$, despite being homeomorphic.

While the coarse perspective originally arose in the context of geometric group theory and index theory of pseudo-elliptic operators on non-compact manifolds \cite{MR2986138, MR1147350, MR4411373}, it has recently come to play a role in mathematical physics in the theory of disordered topological insulators \cite{ludewig2023coarsegeometryapplicationssolid, MR3594362, EM, MR4171428, MR4287182, MR4496385, MR4485912}. In this paper, we will make use of Roe's framework  of \textit{abstract} coarse structures, which provide a metric-independent characterization of coarse geometry in the sense of \cite{MR2007488}, to propose an answer to Question \ref{question1}.

If $X$ is a set, a coarse structure $\mathcal{E}$ is a set of subsets of $X\times X$, i.e. $\mathcal{E}\subseteq \mathcal{P}(X\times X)$ , satisfying several closure axioms (see Definition \ref{Coarse structure}). The elements $E\in \mathcal{E}$ are called \textit{controlled sets}. If $(X,d)$ is a metric space, the the associated bounded coarse structure is defined by 

$$\mathcal{E}^{d}:=\{E\subseteq X\times X\ |\ \sup_{(x,y)\in E} d(x,y)<\infty\}.$$ 

For a fixed set $X$, it will be useful for us to think of the collection of all coarse structures on it as a partially ordered set with respect to containment. Now note that two metrics on $X$ induce \textit{the same} coarse structure precisely when they have the same controlled set, and having the same coarse structure will be the notion that is most relevant to use, rather than abstract coarse equivalence.

We can now ask a version Question \ref{question1} as follows: given a state $\phi$ in the thermodynamic limit of a spin system, is there an intrinsic coarse structure on the sites with respect to which it is local?  This question naturally leads us to consider the collection of \textit{all} coarse geometries on the set of sites with respect to which this state is local. We can then single out the \textit{smallest} coarse structure on $X$ with respect to which $\phi$ is local, which we will see below satisfies a nice universal property.

To formalize this mathematically, suppose that we have an abstract\ (classical or quantum) spin system\footnote{in this paper abstract spin systems are characterized by discrete nets of C*-algebras, see Section \ref{sec:discretenets}}, with the sites indexed by some set $X$. To start with, our system is disembodied, in the sense that we do not assume any geometric structure on $X$ is given a-priori. Let $A_{x}$ denote the C*-algebra of observables localized at site $x$. Typically $A_{x}$ is taken to be $M_{d}(\mathbbm{C})$ in the quantum case and $\mathbbm{C}^{d}$ in the classical case. Let $A$ be the C*-algebra of quasi-local observable, which is simply the infinite tensor product $\otimes_{x\in X} A_{x}$. Let $\phi$ be a state on $A$, and set 

$$C_{\phi}(x,y):=\sup_{a\in A_{x},\ b\in A_{y}}\frac{|\phi(ab)-\phi(a)\phi(b)|}{\|a\| \|b\|}.$$

$C_{\phi}(x,y)$ measures the maximal correlation between observables localized at points $x$ and $y$ respectively. $\phi$ is then said to be \textit{local}\footnote{we will consider higher order variations on this definition, see Definition \ref{higher cont dec} } with respect to a coarse structure $\mathcal{E}$ on $X$ if $C_{\phi}$ has \textit{controlled decay} with respect to $\mathcal{E}$, or in other words, for every $\epsilon>0$, there exists some controlled set $E$ such that for all $(x,y)\notin E$, with $C_{\phi}(x,y)<\epsilon$ (see Defintion \ref{controlled dec}). Before we state the main theorem, recall a coarse structure is \textit{connected} if every finite subset of $X\times X$ is controlled (see Definition \ref{coarsestructures}). We now formulate the following theorem, which is a special case of the more general Theorem \ref{univ. cs}.

\medskip

\begin{thm} Let $\phi$ be a state on an (abstract) spin system over the set $X$. Then

\begin{enumerate}
\item 
There is a coarse structure $\mathcal{E}_{\phi}$ on $X$ such that $\phi$ is local with respect to the coarse structure $\mathcal{E}$ if and only if $\mathcal{E}_{\phi}\subseteq \mathcal{E}$.
\item 
There is a \textit{connected} coarse structure $(\mathcal{E}_{\phi})_{con}$ on $X$ such that $\phi$ is local with respect to the connected coarse structure $\mathcal{E}$ if and only if $(\mathcal{E}_{\phi})_{con}\subseteq \mathcal{E}$.
\end{enumerate}
\end{thm}

\noindent $(\mathcal{E}_{\phi})_{con}$ in the above Theorem can be viewed as a ``connected completion" of $\mathcal{E}_{\phi}$ generated by adding all finite subsets of $X\times X$ to $\mathcal{E}_{\phi}$. 

The above theorem gives a complete characterization of (connected) coarse structures with respect to which $\phi$ is local, and gives natural candidates for a choice of ``universal coarse geometry". There are other natural variations of $\mathcal{E}_{\phi}$ that are similarly universal but with additional hypotheses. The main variation we consider in this paper involves \textit{higher order} controlled decay, which gives rise to the universal coarse structure $\mathcal{E}_{\widetilde{\phi}}$ (see Definition \ref{higher cont dec} and Theorem \ref{scaledef}). In this case, instead of considering operators localized at points, we consider operators localized in uniformly bounded regions. For resonable nice states, we would expect all these variations to agree. The reason we consider all these variations here is that they make the statement of some of our theorems more convenient.  

\subsection{Stability under quasi-local perturbation}. For a coarse geometry associated to a state to be meaningful as a basis for defining order parameters robust under local perturbation, the geometry itself should be robust under local perturbation. Since we are working in the thermodynamic limit, the more natural constraint is to require that the geometry is stable under \textit{quasi-local perturbation}. For the purposes of this paper, a quasi-local pertubation is defined to be quantum channel $\Phi$ on the quasi-local algebra which is a norm limit of quantum channels that are localized in finite regions (see Definition \ref{quasilocalperturbation}). For the case of a pure state on a quantum spin system, stability under quasi-local perturbation is neccessary and sufficient for a property to depend only on its quasi-equivalence class \cite{MR887100}.

To state one of our main results, recall that a coarse structure on a countable set $X$ is \textit{proper} if all bounded sets are finite, and monogenic if it can be generated by a single controlled set. The prototypical examples of proper, monogenic spaces are the coarse spaces associated to infinite, locally finite graphs, e.g. lattices $\mathbbm{Z}^{n}$ and more generally, Cayley graphs of groups. It turns out that the monogenic property of our universal coarse structure $\mathcal{E}_{\phi}$ is equivalent to the coarse scale invariance of the correlations (see Proposition \ref{coarsegeodesic}). One of our main results is the following theorem, which in the body of the paper is proved in a more general setting, see Theorem \ref{full stability}.

 \begin{thm} Let $\phi$ be a state in the thermodynamic limit of a spin system such that $\mathcal{E}_{\phi}$ is proper and monogenic. Then for any quasi-local perturbation $\Psi$, $(\mathcal{E}_{\phi\circ \Psi})_{con}=(\mathcal{E}_{\phi})_{con}$. 
\end{thm} 

There are two immediate consequences for a pure state $\phi$ in the thermodyamic limit of a spin system with $\mathcal{E}_{\phi}$ proper and monogenic.

\bigskip

\begin{enumerate}
    \item 
    (Corollary \ref{normclose}). If $\psi$ is any other pure state with $\||\phi-\psi\||<2$, $(\mathcal{E}_{\phi})_{con}=(\mathcal{E}_{\psi})_{con}$.

    \bigskip
    
    \item 
    If $(\mathcal{E}_{\phi})_{con}$ is a non-trival connected coarse structure, $\phi$ is not a quasi-local perturbation of a product state.
\end{enumerate}

These results show that, at least under the reasonable assumption that $\mathcal{E}_{\phi}$ is proper and monogenic, $(\mathcal{E}_{\phi})_{con}$ is robust against all local perturbations, making it a solid candidate for a universal order parameter.

\subsection{Universal coarse geometry of circuits.} Our discussion so far has focused on studying coarse geometry of states defined in terms of correlation decay, but in fact, our mathematical formalism for constructing universal coarse structures on a set $X$ is very general. The input is essentially any function $f:X\times X\rightarrow\mathbbm{R}_{+}$ (or, for the higher order setting, a function from the product of power sets $\widetilde{f}: \mathcal{P}(X)\times \mathcal{P}(X)\rightarrow \mathbbm{R}_{+}$, see Section \ref{sec 4}). In the above discussion we utilized the two-point correlation function $C_{\phi}$, but there are many other interesting functions related to physical objects that one could consider that we would expect to decay with distance.

An example we find particularly interesting are coarse structures arising from dynamics of quantum spin systems. A (discrete) time evolution on a quantum spin system is local with respect to a given geometry if its ``tails decay" with distance. More formally, using the same setup as before, if $\alpha$ is an automorphism of (or more generally, a unital completely positive (ucp) map on) the quasi-local C*-algebra $A$, then we consider the tail function

$$Q_{\alpha}(x,y):=\sup_{a\in A_{x}, b\in A_{y}} \frac{||[\alpha(a),b]||}{\|a\| \|b\|}.$$

We define the quantum dynamical coarse structure on $X$, $\mathcal{E}_{\alpha}$, to be the universal coarse structure such that $Q_{\alpha}$ has controlled decay. For dynamics arising from integrating gapped, local Hamitlonians or from quantum cellular automata, this construction generically recovers the original coarse geometry (see Section \ref{subsection 4.3} for futher discussion). Unlike the correlation coarse structure, which applies equally well to classical and quantum states, the dynamical coarse structure is only interesting for quantum systems for which localization of observables is detected by commutativity. Nevertheless, we investigate the quantum dynamical structure in parallel with the correlation coarse structure, since many results are completely analogous. Similarly to the correlation case, we can conisder the ``higher order" version $\mathcal{E}_{\widetilde{\alpha}}$. 

The quantum dynamical coarse structure gives an interesting intrinsic coarse geometry for quantum circuits in the thermodynamic limit, and thus the discussion above motivating the study of correlation coarse structures for state realization applies in parallel to universal properties of dynamics.

We have the following stability result precisely parallel to the case of correlation coarse structures (see Theorem \ref{thm: dynamicalcoarsestability}).

\begin{thm} Let $A$ be an (abstract) spin system over the set $X$, and $\alpha:A\rightarrow A$ be any ucp map such that $\mathcal{E}_{\alpha}$ is proper and monogenic. Then for any quasi-local perturbation  $\Psi:A\rightarrow A$, $(\mathcal{E}_{\alpha\circ \Psi})_{con}=(\mathcal{E}_{\alpha})_{con}$. 
\end{thm}

There are two immediate consequences for unitary dynamics (or in other words, when the channel $\alpha$ is an \textit{automorphism of the quasi-local algebra}) on a \textit{concrete} spin system $A$, when $\mathcal{E}_{\alpha}$ is proper and monogenic.

\bigskip

\begin{enumerate}
\item (Corollary \ref{cor:dynamicalcorrollary}). If $\beta\in \text{Aut}(A)$ with $\|\alpha-\beta\|<1$, then $(\mathcal{E}_{\alpha})_{\text{con}}=(\mathcal{E}_{\beta})_{\text{con}}$.

\bigskip

\item 
(c.f Proposition \ref{prop:inner auto}). If $(\mathcal{E}_{\alpha})_{\text{con}}$ is a non-trivial connected coarse structure, $\alpha$ is not an inner automorphism.
\end{enumerate}

\bigskip

\noindent \textbf{State preparation}. In Section \ref{loc pres chan}, we tie (higher-order) correlation and dynamical coarse structures together in the context of state preparation on a quantum processor. We study the behavior of coarse structures of states under locality preserving channels. As a consequence of this analysis, we obtain the following as a direct consequence of Corollary \ref{dynamics and states}. 

\begin{thm}
 Suppose $\phi$ is a state and $\alpha$ is a circuit that prepares $\phi$ from a product state. Then $\mathcal{E}_{\widetilde{\phi}}\subseteq \mathcal{E}_{\widetilde{\alpha}}$    
\end{thm}

If we interpret the coarse structures $\mathcal{E}_{\widetilde{\phi}}$ and $\mathcal{E}_{\widetilde{\alpha}}$ as a sort of ``coarse complexity class", this result says any circuit used to prepare a state must be at least as coarsely complex as the state. For example, if we want to prepare a state whose universal (higher-order) coarse structure is equivalent to $\mathbbm{Z}^{3}$ on a quantum processor whose hardware is coarsely equivalent to $\mathbbm{Z}^{2}$, we can not use a circuit which is (higher-order) local with respect to the hardware geometry.

\subsection{Coarse dependence of order parameters} In the final section of the paper, we consider the observation that many entanglement-based order parameters of a state, such as topological order, depend explicitly on some kind of geometry to be defined. We have proposed (a family of) candidates for canonical coarse geometries that need not be compatible with spatial geometries directly, but we need to work out the extent that order parameters of interest only depend on coarse geometries in the first place. This is the topic of Section \ref{coarsedependence}.  

Here, we work with an explicit metric $d$ on the set of sites $X$, and want to show various order parameters are invariant under coarse equivalence. We make the standing assumption that the spaces are monogenic, which allows us to reduce coarse equivalence to quasi-equivalence (see Proposition \ref{coarsemono}). Under this assumption on $d$, $d$ is coarsely equivalent to another metric $d^{\prime}$ if and only if there are constants $L>0, C\ge 0$ such that for any $x,y\in X$,

$$L^{-1}d(x,y)-C\le d'(x,y)\le Ld(x,y)+C$$

The goal, then, is to see which definitions of various order parameters, depending a-priori on a metric $d$, are invariant under this equiavalence relation.

\medskip

\noindent \textbf{Superselection sectors and topological order}. The first order parameter we investigate in Section \ref{topord} is the W*-category of cone-localized superselection sectors of a state $\phi$ defined on a spin system in $\mathbbm{Z}^{n}$. In 2-dimensions, under some assumptions this category can be upgraded to a braided tensor category, which characterizes the states' \textit{topological order} \cite{MR2804555, MR4362722}. For $(X,d)$ having finite asymptotic dimension (see Definition \ref{asymptoticdim}), we give a coarse variation of the definition of superselction sectors in terms of coarse cones, and we show that this W*-category only depends on the coarse structure. We also show that in the case that the space is $\mathbbm{Z}^{n}$ with the Euclidean metric, this agrees with the usual definition of superselction sectors in terms of localization in honest cones (see Proposition \ref{topordprop}), hence this W*-category is a coarse order parameter.

\bigskip

\noindent \textbf{Decay of correlations}. Another important family of order parameters are related to the decay \textit{rate} of correlation functions. In \ref{decay rates}, we investigate to what extent the decay rate depends only on the coarse structure of the metric. Generic decay rates are not always coarsely defined. 

We first investigate the three most common classes of decay that are commonly considered in the condensed matter literature: 0 correlation length (correlations are $0$ when distance is large), finite correlation length (correlations have at most exponential decay), and infinite correlation length (correlations have super-exponential decay). 

\begin{prop}(c.f. Proposition \ref{coarseinvariancecorrelations}).
The definitions correlation length $0$, finite correlation length, and infinite correlation length for a state $\phi$ each depend only on the quasi-equivalence class of the metric. Thus if $\mathcal{E}_{\phi}$ is monogenic, then having $0$, finite, or infinite correlation length is a coarse order parameter.
\end{prop}

While the property of having finite correlation length is coarsely well defined, the actual value of the correlation length is not, which leads us to believe that generally the precise asymptotic ``coarse rate of decay" for correlations is not something that would make sense for correlation functions.

However, an important family of order parameters are \textit{critical exponents}, which help determine the scaling limit CFT that is believed to characetrize the universality class of critical states of spin systems (see, for example \cite{stanley1987introduction}). These numbers are typically defined in terms of degrees of ploynomials for parameters of the system with algebraic decay. In our setting, we (unsurprisingly) focus on algebraic decay of correlations. In section \ref{decay rates}, we develop a bit of theory that allows us to make sense of when a class of decay functions on $\mathbbm{R}^{+}$ are asymptotically invariant under quasi-equivalent changes of the metric, introducing the notion of \textit{coarsely stable}. As a consequence of this theory we are able to conclude the following.

\begin{prop}(See Proposition \ref{criticalexp-welldefined}). Algebraic decay of correlations and the critical exponent are independent of the quasi-equivalence class of the metric.
\end{prop}

To us, this result strongly suggests that what we think of as universal features of a system, typically defined as properties of a scaling limit QFT, in fact make sense and could be defined at the coarse level instead, which might be more accessible from a mathematically rigorous viewpoint. This line of reasoning suggests that it may be of interest to develop a notion of ``coarse equivalence" of quantum field theories (for example, the framework of algebraic quantum field theories), and instead of defining a universality class in terms of renormalization group fixed points, define a universality class as a coarse equivalence class of quantum field theories obtained in the thermodynamic limit (possibly with some kind of stabilization).

\subsection{Discussion.}To put our perspective in context, we note that we were inspired by  similar ideas for extracting geometric structures from states and/or dynamics that have already appeared in the literature to describe \textit{emergent spacetime} (for example, see \cite{qi2013exactholographicmappingemergent, MR3773571, MR3600376} and references therein). These works have somewhat different goals, mostly focusing on (small-scale) metric geometry with the aim of obtaining something that looks like gravity on their emergent spatial structure. This contrasts with the goals of our project, which focuses on large-scale geometry emerging only in the thermodynamic limit of a spin system, for the purpose of serving as a basis for universality classes. However, these ideas use essentially the same mathematical starting point of using quantum information as a basis for defining geometry. Our program is also related to \cite{Qi2019determininglocal}, where the authors similarly utilize the correlation functions with the aim of reconstructing Hamiltonians from states.

Another aspect of our analysis we wish to emphasize is that the choice of the ``sites" is taken as given for us. From the perspective of state realization on a quantum processor, this is a natural starting point since these machines have a built-in choice of sites and local observables. In the abstract setting, one could wonder if we could derive a site structure itself from just a state on an algebra. This seems unlikely, since any two states in the thermodynamic limit of a spin system are related by an automorphism, but it would be very interesting to see if some natural additional structure, such as a Hamiltonian, could give rise to a decomposition into sites  

A natural follow-up question is whether we can find interesting order parameters that are valued in algebraic structure related to coarse spaces. For example, coarse spaces have (co)-homology and K-theories, and it would be interesting to work out if states have invariants valued in these groups, as in the case of disordered band insulators \cite{EM}. Another interesting issue is the question of actually computing aspects of the emergent coarse structure in practice, either theoretically or from a large but finite data set (See section \ref{asymptoticdimgrowthrate}). This problem has many similarities to the problems addressed in topological data anlysis \cite{10.3389/frai.2021.667963} and in particular, persistance homology which we hope to explore in the future.

Finally, we remark on our mathematical formalism, which some readers may find excessively general. Although we have stated our results above for quantum spin systems, all our definitions and theorems are formulated in the language of abstract operator algebras, since proofs work at this level of generality. One possible advantage is that our framework applies equally well to classical systems with probability distribution states (i.e. ``macrostates") by assuming the algebras in question are commutative, as well as more complicated algebras arising from quantum field theory and topological order. In addition, we note that our framework for extracting coarse structures on the set $X$ is also very general and may be of independnet interest. Indeed, it only requires as input a function $f:X\times X\rightarrow \mathbbm{R}_{+}$, and thus many other objects relevant in physics (e.g. a Hamiltonian) or even in classical network theory could be used to cook up a function $f$ to produce a relevant coarse structure.

\subsection{Structure of the paper} The outline of the paper is as follows. Section 2 reviews the C*-algebraic formulation of states, observables, and dynamics, and introduces the framework of discrete nets of algebras used throughout the paper. Section 3 gives a mostly self-contained review of the important basic concepts in coarse geometry. Section 4 is the most substantial section. We introduce our formalism for extracting coarse geometries on a set $X$ from functions $f:X\times X\rightarrow \mathbbm{R}$, and discuss our main examples arising from correlation functions and dynamics. In section 5, we prove our stability results for both correlation and dynamical structures. In section 6, we investigate the interaction between dynamical and correlation coarse structure. Finally in Section 7 we discuss the dependence of superselection sectors and decay rates on the coarse structure.

\subsection{Acknowledgements}

We would like to thank Pieter Naaijkens, David Penneys and Dominic Williamson for interesting and useful discussions around these topics. Both authors were supported by NSF DMS- 2247202.

\section{C*-algebra formalism}

The mathematical setup for this paper uses the C*-algebra framework for describing physical systems. We briefly review this here, but refer the reader to \cite{MR887100, MR3643288} for all definitions and results we mention below concerning C*-algebras, as well as a more extensive overview of their role in mathematical physics.

The basic idea is that a mathematical description of any sort of dynamical system requires specifying three basic ingredients and how they interact: observables (the kind of measurements you can make, in principle), states (possible configurations of the system, or perhaps possible states of knowledge of the configurations), and dynamics (how the states and/or observables change in time). Different kinds of theories (or even different formulations of the same theory) may postulate different flavors of mathematical structures to instantiate these concepts, and there may even be different interpretations of the same mathematical structure. 

In the C*-algebraic approach to physical systems, the structure of the three ingredients is determined by a C*-algebra $A$. Recall a C*-algebra $A$ is a  complex Banach $*$-algebra, whose norm satisfies the C*-condition $||a^{*}a||=||a||^{2}$. The standard examples of (unital) C*-algebras are the algebra $B(H)$ of all bounded operators on the Hilbert space $H$, which for an $n$-dimensional Hilbert space is isomorphic to $M_{n}(\mathbbm{C})$. Commutative examples are $C(X)$, the algebra of complex-valued continuous functions on a compact Hausdorff space $X$, with supremum norm and pointwise multiplication and addition. In fact, by Gelfand duality this class comprises precisely the commutative unital C*-algebras. The commutative algebras can be used to formulate classical theories, while noncommutativity gives rise to quantum phenomena.

Algebraic quantum theories postulate the following:

\begin{enumerate}
\item 
\textbf{Observables} of the theory are the self-adjoint elements of a C*-algebra $A$.
\item 
\textbf{States} are bounded linear functionals $\phi:A\rightarrow \mathbbm{C}$ such that $||\phi||=1$ and $\phi(x^{*}x)\ge 0$.
\item 
\textbf{Dynamics} are $*$-automorphisms\footnote{ or more generally, unital completely positive (ucp) maps for open dynamics} $\alpha\in \text{Aut}(A)$,
\end{enumerate}

The set of states $S(A)$ forms a convex space, and we are often interested in \textit{pure} states, which are the extreme points in $S(A)$. These are sometimes considered the ``true" configurations of the system, since they contain maximal information in a certain sense, while a general state contains some (classical) uncertainty. Given a state $\phi\in S(A)$ and an observable $a\in A$, the value $\phi(a)$ is interpreted as the expected value of the outcome of measuring the observable $a$ when the system is in the state $\phi$.  For an abelian C*-algebra $A\cong C(X)$, the pure states correspond to evaluations at a point $x\in X$, while arbitrary states correspond to Radon measures on $X$. Applied to the C*-algebra of compact operators $K(H)$ on the Hilbert space $H$, the state space can be identified with density matrices $\rho$ via the correspondence $\rho\leftrightarrow \text{Tr}(\cdot \rho)$. This recovers the standard quantum mechanical picture of states as density matrices, where states are given by tracing against density matrices.

A single automorphism $\alpha\in Aut(A)$ is interpreted as a discrete time evolution. If a system is in state $\phi$ at time $0$, then after one time step, it is in state $\phi\circ \alpha$. A continuous time evolution is decribed by a homomorphism $\mathbbm{R}\rightarrow \text{Aut}(A)$, $t\mapsto \alpha_{t}$. In this case, if a system is initially in state $\phi$ at time $t=0$, then at time $t$ the state is $\phi\circ \alpha_{t}$. If the homomorphism is continuous in the point norm topology, $\{\alpha_{t}\}$ has an infinitesimal generator $\delta$ which is a (generally unbounded) derivation on the C*-algebra $A$. This can be identified as the Hamiltonian of the system, and the equation expressing that time evolution is generated by $\delta$ is a generalization of the Schrodinger equation (for further details, see \cite{MR887100}).

More generally, we may consider \textit{open dynamics}, which are given instead by \textit{quantum channels}. These are expressed in the Heisenberg picture as unital, completely positive maps $\alpha: A\rightarrow A$, and the action on states is $\phi \mapsto \phi\circ \alpha$ (see \cite{MR1976867}). Clearly, automorphisms (which we will call ``unitary dynamics") give examples of quantum channels. However, channels actually simultaneously generalize states and dynamics, since any state also defines a channel, via the map $a\mapsto \phi(a)1_{A}$. In this paper, we will be thinking of a unital completely positive maps mostly as generalized dynamics.

\bigskip

\subsection{Discrete nets of C*-algebras}\label{sec:discretenets}. The fundamental example of a quantum system where the algebraic framework can be very useful is the thermodynamic limit of a \textit{spin system} \cite{MR1441540}. Here, we start with a (usually countably infinite) set of sites $X$. At each site $x\in X$ we have a unital C*-algebra $A_{x}$ of localized observables. In this paper, we will assume this algebra is nuclear, so that we can ignore the technicalities arises from the different types of tensor products. The examples we have in mind are spin systems, so that $A_{x}\cong M_{n}(\mathbbm{C}))$ or $\mathbbm{C}^{n}$ for some $n$. However, all the math works virtually unchanged more generally.

For any finite subset $F\subseteq X$, we define the  C*-algebra

$$A_{F}:=\otimes_{x\in F} M_{d}(\mathbbm{C}).$$

\noindent This means we take the tensor product algebra whose tensor factors are explicitly indexed by the elements of $F$. For $F\subseteq G$, we have a natural inclusion $A_{F}\hookrightarrow A_{G}$, given by $a\mapsto 1^{\otimes G-F}\otimes a$. We define the C*-algebra

$$A=\text{colim}_{F\in \mathcal{F}(X)} A_{F},$$

\noindent where $\mathcal{F}(X)$ denotes the partially ordered set of finite subsets, and the colimit is taken in the category of C*-algebras. For the less categorically inclined, if we identify each $A_{F}$ with its image in $A_{G}$ for $F\subseteq G$, then this colimit is just the completion of the union algebra (for more technical details, see \cite{MR1783408}). Identifying each $A_{F}$ with its image in $A$ equips $A$ with a distinguished family of subalgebras indexed by finite regions of $X$. $A$, taken together with this family of subalgebras, provides a mathematically precise description of the local observables of the system. For two disjoint regions $F\cap G=\varnothing$, $[A_{F},A_{G}]=0$. This results in a connection between the algebraic relations of the observables and their ``location in space".

While spin systems are the main example we consider, most of the definitions we provide make sense (and may be of interest) in greater generality, in particular for nets of symmetric operators in an ordinary spin system with a global gauge symmetry. The following definition, which abstracts some of the main features of spin systems we described above, has been used recently in the study of symmetric quantum cellular automata and in the theory of local topological order \cite{Jones2023DHRBO}, \cite{JL24, jones2023localtopologicalorderboundary, tomba2023boundaryalgebraskitaevquantum}. It is essentially a discrete metric version of an algebraic quantum field theory (AQFT) in the sense of Haag and Kastler \cite{MR165864, MR1405610}, but without a distinguished Hilbert space representation or symmetry group.

\begin{defn}
A \textit{discrete net of C*-algebras} consists of a countable set $X$, a unital C*-algebra $A$ called the \textit{quasi-local algebra}, and for each finite subset $F\subseteq X$, a unital subalgebra $A_{F}\subseteq A$ (called a local subalgebra) such that

\begin{enumerate}
\item
If $F\subseteq G$, then $A_{F}\subseteq A_{G}$.
\item
IF $F\cap G=\varnothing$, then $[A_{F},A_{G}]=\varnothing$.
\item 
$\bigcup_{F\in\mathcal{F}(X)} A_{F}$ is norm dense in $A$.
\end{enumerate}
\end{defn}

Given \textit{any} subset $F\subseteq X$ (not neccessarily finite), we can define $A_{F}$ to be the C*-algebra generated by the collection of $A_{G}$ where $G\subseteq F$ is finite, allowing us to make sense of the \textit{quasi-local observables} localized in any arbitrary region $F$. Thus we have an assignment of a subalgebra of $A$ to arbitrary subsets $\mathcal{P}(X)$ of $X$. 

Secondly, in practice, $X$ typically has the structure of a metric space (e.g. the metric space $\mathbbm{Z}^{d}\subseteq \mathbbm{R}^{d}$), and one usually asks their nets to satisfy additional assumptions that tie the algebra to the geometry (for example, something like algebraic Haag duality in \cite{Jones2023DHRBO}). However, in this paper, the whole point is that there is not a pre-existing geometry, and the purpose is to \textit{derive} a geometric structure from correlations or dynamics of the system.

Thirdly, there is nothing fundamental about the unital assumption of $A$ or the local algebras $A_{F}$, and if we prefer to use non-unital C*-algebras, we should likely require that the inclusions are non-degenerate in some sense. However, the examples we have in mind (spin systems and related examples) all have the property that the local algebras $A_{F}$ are finite-dimensional, which are automatically unital. For locally infinite dimensional nets, it may be more useful to require the additional hypothesis that the local algebras are in fact von Neumann algebras (which is automatic in the locally finite dimensional case), though this makes no difference to this discussion for our purposes.

If we are thinking of a discrete net of C*-algebras as generalizing a spin system, the major difference (aside from the type of the local C*-algebra) is witnessed by \textit{algebraic entanglement}. For a general discrete net $A$, for $F\cap G=\varnothing$, our inclusions gives a natural map $A_{F}\otimes A_{G}\mapsto A_{F\cup G}$. In an ordinary spin system, this is an isomorphism of algebras. For a general net, this may fail to be injective or surjective. In other words, products of operators localized in disjoint regions may be $0$, and local operators may not factorize as sums of products of operators localized in $F$ and $G$  respectively. Physically, this latter condition implies that local observables are not completely determined by the observables localized at points. See Section \ref{OtherCoarse} for further discussion. Below we give some naturally occurring examples of discrete nets beyond ordinary spin systems.

\bigskip

\begin{enumerate}
\item
\textbf{Symmetric spin systems}. Let $A:=\otimes_{x\in X} B$ be a generalized spin system, and suppose we have an action of the group $G$ on the local C*-algebra $B$. Then the diagonal action on $A$ is called an \textit{on-site global symmetry}. We can consider the C*-algebra $A^{G}$ of $G$ invariant operators, and the local operators $A^{G}_{F}$ is just $A^{G}\cap A_{F}$. These systems generically exhibit algebraic entanglement.
\item 
\textbf{Categorical spin systems}. Associated to higher categorical structures (e.g. unitary fusion categories, unitary braided fusion categories, etc), one can construct nets of finite dimensional C*-algebras. The most studied of these are \textit{fusion spin chains} \cite{Jones2023DHRBO}. These arise either as the local invariant operators under a (weak) Hopf or categorical symmetry, but also as the boundary algebras of 2+1D topological orders \cite{jones2023localtopologicalorderboundary}.
\item
\textbf{AQFTs}. A (Euclidean) algebraic quantum field theory is a net of C*-algebras $B$ over a (complete) Riemannian manifold $M$, assigning C*-algebras to (bounded) open sets (c.f. \cite{doi:10.1142/S0129055X99000362}). Take a geodesic triangulation $S$ of $M$ such that the vertices are at at least $\delta$ distance apart for some fixed $\delta>0$. If we let $X$ be the set of vertices of the triangulation, we can define a discrete net of C*-algebras $A$ over $X$ by defining, for any finite subset $F\subseteq X$, $A_{F}:=B_{N_{\frac{\delta}{2}}(S(F))}$, where$S(F)$ denotes the subcomplex spanned by the vertices in $F$. This is a ``discretization" of the original theory.
\end{enumerate}

\section{Coarse geometry}\label{coarsestructures}

As explained in the introduction, coarse structures on a set are designed to provide an abstract characterization of large scale geometric structure. They were originally developed by John Roe, emerging from applications in geometric group theory and index theory of Dirac operators on manifolds. A coarse structure is very a flexible notion that captures our intuition for large scale equivalence, ignoring the small scale structure of the geometry. For example, $\mathbbm{R}^{n}$ with the usual Euclidean metric and $\mathbbm{Z}^{n}$ with the taxi-cab metric are coarsely equivalent, despite not even having the same cardinality as sets! For a comprehensive reference on coarse geometry, we refer to the textbook of Roe \cite{MR2007488}. In this section, we brielfy recall some of the basic terminology and constructions we will need.

\begin{defn}\label{Coarse structure} Let $X$ be a set. A \textit{coarse structure} on $X$ is a collection $\mathcal{E}$ of subsets of $X\times X$ satisfying the following conditions:

\begin{enumerate}
    \item 
    $\mathcal{E}$ contains the diagonal, and is closed under taking subsets and finite unions.
    \item
    If $E\in\mathcal{E}$, then $E^{-1}=\{(x,y)\ :\ (y,x)\in E\}\in \mathcal{E}$
    \item
    If $E,F\in \mathcal{E}$, then $E\circ F:=\{(x,z)\ : (x,y)\in E, \text{and}\ (y,z)\in F\ \text{for some y}\}\in \mathcal{E}$.
\end{enumerate}

\end{defn}

The sets $E\in \mathcal{E}$ are called \textit{controlled}. A set $B\subseteq X$ is called \textit{bounded} if $B\times B\in \mathcal{E}$. The motivation for this terminology arises from the following motivating example, which was discussed in the introduction.

\begin{ex}(Bounded coarse structure of a metric space). If $X$ is a set, recall that a metric on $X$ is a function $d:X\times X\rightarrow \mathbbm{R}_{+}\cup \{\infty\}$ which is symmetric, definite, and satisfies the triangle inequality. Define a set $E\subseteq X\times X$ to be controlled if $$\displaystyle \sup_{(x,y)\in E} d(x,y)<\infty.$$ The resulting coarse structure, which we denote $\mathcal{E}^{d}$ is called the \textit{bounded coarse structure} or also simply the \textit{metric coarse structure} of $(X,d)$.
\end{ex}

The idea is that the bounded coarse structure of a metric space captures its large-scale behavior, by telling us which collection of pairs of points are a uniformly bounded distance apart as they go off to infinity. This is in some sense an inversion of the uniform structure of a metric space, which is similarly defined but capture the small-scale structure of a metric.

We call a coarse space $(X,\mathcal{E})$ \textit{metrizable} if there is some metric $d$ on $X$ whose bounded coarse structure is $\mathcal{E}$. By \cite{MR2007488}, a coarse structure is metrizable if and only if it is \textit{countably generated}. All the coarse structures that occur naturally in this paper will be metrizable.

For any set $B\subseteq X$ and any controlled set $E\in \mathcal{E}$, the $E$-ball $E[B]:=\{x\in E\ : (y,x)\in E\ \text{for some y}\in B\}$. A (discrete) coarse structure is called \textit{proper} if the $E$-ball of any bounded set is finite. Note that this implies all bounded sets are finite. This has a generalization to coarse structures on spaces that already have an underlying topology, where we require E-balls of bounded sets to be compact, but here we will only be using the ``discrete" version described above.

A coarse space $(X,\mathcal{E})$ is \textit{coarsely connected} if it contains the collection of finite subsets $\mathcal{F}(X\times X)$. Equivalently, a coarse space is coarsely connected if every pair of points is contained in a controlled set. Thus we see that for a metric space $(X,d)$, $\mathcal{E}^{d}$ is coarsely connected if and only if the metric takes finite values for any pair of points.

A \textit{coarse function} between coarse spaces $(X,\mathcal{E})$ and $(Y, \mathcal{F})$ is a function $\rho:X\rightarrow Y$ such that the inverse image of bounded sets are bounded, and $(\rho\times \rho)(\mathcal{E})\subseteq \mathcal{F}$ (the image of controlled sets are controlled). Two coarse functions $\rho, \gamma:(X,\mathcal{E})\rightarrow (Y, \mathcal{F})$ are \textit{close} if $\{(\rho(x),\gamma(x))\ :\ x\in X\}\in \mathcal{F}$. Closeness is an equivalence relation on the set of coarse functions. The \textit{coarse category} consists of coarse spaces and coarse maps up to the equivalence relation of closeness, with usual composition of functions. We say two coarse spaces are \textit{coarsely equivalent} if they are isomorphic in the coarse category. 

\begin{ex}\label{Cayley}
Let $G$ be a finitely generated group. For any finite generating set $S$, we can define the Cayley graph whose vertices are the elements of $G$, with an edge between $h$ and $g$ is there is an $s\in S$ such that $sh=g$. Utilizing the path metric makes $G$ into a discrete metric space. Any two choices of finite generating set yield equivalent coarse structure on $G$. These are the coarse spaces that play such a prominent role in geometric group theory.
\end{ex}

\begin{ex}An illuminating example demonstrating change in cardinality is the following: let $(X, d)$ be a metric space, let $R>0$, and let $Y\subseteq X$ such that for every $x\in X$ there exists a $y\in Y$ with $d(x,y)<R$. Then $Y$ is coarsely equivalent to $X$. Thus any separable metric space (e.g. $\mathbbm{R}^{n}$) is coarsely equivalent to a discrete metric space.
\end{ex}

In this paper, we will mostly be studying different coarse structures on the same set and whether or not they are contained in each other, which is slightly different from the usual perspective in coarse geometry. Indeed, our function comparing coarse structures is always the identity, and it may not be coarse in situations that we might expect. Indeed, if $\mathcal{E}$ and $\mathcal{F}$ are two coarse structures on $X$ with $\mathcal{E}\subseteq \mathcal{F}$, it is not necessarily the case that the identity is a coarse function from $(X,\mathcal{E})\rightarrow (X,\mathcal{F})$, since $\mathcal{F}$-bounded sets are not neccessarily $\mathcal{E}$-bounded. 

If $\mathcal{E}$ and $\mathcal{F}$ are coarse structures on $X$, and $\mathcal{E}\subseteq \mathcal{F}$ we say that $\mathcal{E}$ is \textit{smaller} than $\mathcal{F}$ or that $\mathcal{F}$ is \textit{larger} than $\mathcal{E}$. There is a smallest coarse structure, the \textit{discrete} coarse structure, which consists only of all subsets of the diagonal. There is a largest coarse structure, the \textit{indiscrete} coarse structure, which consists of all subsets of $X\times X$. The discrete coarse structure is the bounded coarse structure of a metric which gives an infinite value between any pair of distinct points, while the indiscrete coarse structure is given by any metric which is globally bounded. The coarse structure $\mathcal{F}(X\times X)$ of finite subsets is the smallest connected coarse structure, and we will frequently call this the trivial connected coarse structure.

We have the following easy proposition, which gives us a useful tool for comparing coarse structures in the metrizable case.

\begin{prop}\label{comparison}
Let $d$ and $d^{\prime}$ be metrics on $X$. Then $\mathcal{E}^{d}\subseteq \mathcal{E}^{d^{\prime}}$ if and only if there is some non-decreasing function $\rho:\mathbbm{R}_{+}\rightarrow \mathbbm{R}_{+}$ with 
$$d^{\prime}(x,y)\le \rho(d(x,y)),$$

\noindent where we interpret $\rho(\infty)=\infty$.
\end{prop}

\begin{proof}
Suppoe $\mathcal{E}^{d}\subseteq \mathcal{E}^{d^{\prime}}$. Then for every $R\ge 0$, $\{(x,y)\in X\times X\ :\ d(x,y)\le R\}\in \mathcal{E}^{d^{\prime}}$, we can choose some $S_{R}$ such that $d^{\prime}(x,y)\le S_{R}$. Since we are free to choose arbitrarily large $S_{R}$, we can choose $S_{R}$ so that the function $\rho(R):=S_{R}$ is increasing. This satisfies the criterion. Conversely, given such a $\rho$, let $E\in \mathcal{E}^{d}$, and choose an $R\ge 0$ such that $d(x,y)\le R$ for all $(x,y)\in E$. Then $d^{\prime}(x,y)\le \rho(d(x,y))\le \rho(R)$, thus $E\in \mathcal{E}^{d^{\prime}}$.
\end{proof}

The intersection of any family of coarse structures is again coarse. Thus if $\mathcal{S}$ is any family of subsets of $X\times X$, there is a smallest coarse structure containing $\mathcal{S}$, which we call the coarse structure generated by $\mathcal{S}$ and denote by $\langle \mathcal{S}\rangle$. This coarse structure is obtained as the intersection of all coarse structures containing it (this is non-empty since every family $\mathcal{S}$ is contained in the indiscrete coarse structure).  We have the following observation:

\begin{prop}\label{genra}
If $\{\mathcal{E}_{\epsilon}\}_{\epsilon\in I}$ is a family of coarse structures on the set $X$, then

$$\langle \mathcal{E}_{\epsilon}\rangle_{\epsilon \in I}=\{E\subseteq X\times X\ :\ E\subseteq E_{1}\circ \dots E_{n}\ \text{where}\ E_{i}\in \mathcal{E}_{\epsilon_{i}},\ \epsilon_{i}\in I\}$$
\end{prop}

\begin{proof}
    Obviously the collection $\mathcal{S}:=\{E\subseteq X\times X\ :\ E\subseteq E_{1}\circ \dots E_{n}\ \text{where}\ E_{i}\in \mathcal{E}_{\epsilon_{i}},\ \epsilon_{i}\in I\}$ is a coarse structure containing ${\mathcal{E}_{\epsilon}}$, hence contains $\langle \mathcal{E}_{\epsilon}\rangle_{\epsilon \in I}$. Conversely any coarse structure containing all $\mathcal{E}_{\epsilon}$ will neccessarily contain $\mathcal{S}$, hence $\mathcal{S}\subseteq \langle \mathcal{E}_{\epsilon}\rangle_{\epsilon \in I}$.
\end{proof}

\begin{ex}\label{connected quotient} If $(X,\mathcal{E})$ is not coarsely connected, we can always consider its connected quotient $\mathcal{E}_{\text{con}}$, obtained by staking the coarse structure generated by $\mathcal{E}$ and all finite subsets, or in other words $\mathcal{E}_{con}:=\langle \mathcal{E}\cup \mathcal{F}(X\times X)\rangle$.
\end{ex}

\begin{ex}\label{monogenic}
A coarse structure is called \textit{monogenic} if it is generated by a single subset $E\subseteq X\times X$. If $G$ is a graph, then the bounded coarse structure on the vertices $V(G)$ is monogenic with generator $E=\{(x,y)\ : \{x,y\}\in E(G)\}$. In fact, this graph presentation is characteristic for monogenic coarse structures. If $\mathcal{E}=\langle F\rangle$ is a monogenic coarse structure on $X$, consider the graph $G_{F}$ whose vertices are $X$ and $\{x,y\}\in E(G_{F})$ if either $(x,y)\in F$ or $(y,x)\in F$. Then it is easy to see that $\mathcal{E}$ is the bounded coarse structure induced from the path metric on $G_{F}$, which we denote by $d_{F}$. 

In particular, monogenic coarse structures are precisely metrizable by path metrics on some graphs, and from this we can conclude that they are restrictions of \textit{geodesic} metrics. More generally, one can show that a coarse space is coarsely geodesic if and only if it is monogenic \cite[2.57]{MR2007488}.

Note that in the monogenic case, the coarse structure is coarsely connected if and only if the graph $G_{F}$ is connected, and is proper if and only if $G_{F}$ is locally finite.
\end{ex}

\begin{prop}\label{coarsemono}(c.f. \cite[1.4.14]{MR4646531}) Suppose $(X,\mathcal{E})$ is monogenic with generator $F$, and let $d_{F}$ be the metric as in Example \ref{monogenic}. Then for any other metric $d^{\prime}$, $\mathcal{E}\subseteq \mathcal{E}^{d^{\prime}}$ if and only if there is some $L>0$ such that $$d^{\prime}(x,y)\le Ld_{F}(x,y).$$

In particular, if $G, F$ are both generators for $\mathcal{E}$, then there is a positive constant $K$ such that

$$\frac{1}{K} d_{G}(x,y)\le d_{F}(x,y)\le K d_{G}(x,y).$$
\end{prop}

\begin{proof}
Clearly $$d^{\prime}(x,y)\le Ld_{F}(x,y)$$ implies $\mathcal{E}\subseteq \mathcal{E}^{d^{\prime}}$ by Proposition \ref{comparison}. In the other direction, by Proposition \ref{comparison} there exists some non-decreasing $\rho:\mathbbm{R}_{+}\rightarrow \mathbbm{R}_{+}$ such that $d^{\prime}(x,y)\le \rho(d_{F}(x,y))$. Let $(x=x_{0}, x_{1}, \dots,\ x_{n}=y)$ be a path in $G_{F}$ so that $d_{F}(x,y)=n$. Then 

$$d^{\prime}(x,y)\le \sum_{i} d^{\prime}(x_{i},x_{i+1})\le \sum_{i} \rho(d_{F}(x_{i},x_{i+1}))=n\rho(1)=\rho(1)d_{F}(x,y)$$
\end{proof}

This says that for monogenic coarse structures, there is in fact a canonical \textit{Lipschitz equivalence class} of metric, which is actually stronger than coarse equivalence. However, in this paper we are more interested in the coarse equivalence class of the metric, since ultimately it is the underlying coarse structure we will be studying.

It turns out that the metric $d_{F}$ arising from a monogenic coarse structure is an example of a \textit{quasi-geodesic metric space}, and for these, coarse equivalence reduces to the easier to work with notion of quasi-isometry. To formalize this we have the following definition.

\begin{defn}(\cite[1.4.10]{MR4646531})
    A metric space $(X,d)$ is quasi-geodesic if there exists an $L, C>0$ such that for any $x,y\in X$ with $d(x,y)=t<\infty$, there exists a function $\gamma: [0,t]\rightarrow X$ such that $\gamma(0)=x$, $\gamma(t)=y$, and for all $0\le s,s^{\prime}\le t$ $$L^{-1}d(\gamma(s),\gamma(s^{\prime}))-C\le |s^{\prime}-s|\le Ld(\gamma(s),\gamma(s^{\prime}))+C $$
\end{defn}

Clearly if $\mathcal{E}$ is generated by $F$, then $d_{F}$ is quasi-geodesic, so monogenic coarse structures can always be realized by quasi-geodesic metrics. The following proposition makes this useful.

\begin{prop}(\cite[1.4.14]{MR4646531}) If $(X,d)$ is a quasi-geodesic metric space (and in particular if $\mathcal{E}^{d}$ is mongenic) and $d^{\prime}$ is any other coarsely equivalent metric, then there are $L>0, C\ge 0$ such that

$$L^{-1}d(x,y)-C\le d^{\prime}(x,y)\le Ld(x,y)+C.$$

\end{prop}

In general, two metrics $d$ and $d^{\prime}$ that satisfy the above inequalities are called \textit{quasi-isometric} or \textit{large scale Lipschitz equivalent}. The above proposition implies, in particular, that for monogenic coarse spaces $(X,\mathcal{E})$, if we want to consider all metrics which generate this coarse structure it suffices to consider the quasi-isometry class of any particular one.

\section{Coarse structures defined by controlled decay} \label{sec 4}

In this section we will record some general considerations concerning coarse structure defined on countable sets in terms of the notion of a controlled decay of a function $f:X\times X\rightarrow \mathbbm{R}_{+}$. The standing example will be the correlation functions of a (classical or quantum) spin system, which, for the moment, we treat informally. The idea is to take observables \textit{localized} at sites $x$ and $y$ and consider their covariance 

$$\text{cov}(a,b):=\frac{|\phi(ab)-\phi(a)\phi(b)|}{\|a\| \|b\|}$$

This measures the statistical correlation of the observables $a$ and $b$. If this quantity is $0$, it means that the observables $a$ and $b$ are statistically independent, so we can think of the covariance as measuring the deviation of these observables from independence.

We can measure the maximal correlation between observables localized at $x$ and $y$ respectively with the function that takes the supremum of covariances between observables localized at different sites

$$C_{\phi}(x,y):=\sup_{a\in (A_{x}),\ b\in  (A_{y})}\frac{|\phi(ab)-\phi(a)\phi(b)|}{\|a\| \|b\|}.$$

The key point is that for states arising from \textit{local} physical processes with respect to some metric $d$ on $X$, we expect localized observables to become decreasingly correlated as the distance between the sites increases. Indeed, we will operationally use this to \textit{define} the locality of a state with respect to a given metric $d$. As explained in the introduction, however, this notion of locality of a state only depends on the bounded coarse structue of the metric and not on the local details of the metric. This motivates the following mathematical problem: given a state on a discrete net of over $X$ without any a-priori metric, what are all the coarse structures on $X$ with respect to which our state is local?

We approach this problem by abstracting away the physical details of states and correlations, and develop the theory simply given a function $f:X\times X\rightarrow \mathbbm{R}_{+}$. This allows us the flexibility to apply our framework in other settings, in particular the quantum dynamical coarse structures in Section \ref{subsection 4.3}.

\begin{defn}\label{controlled dec}
   Let $(X,\mathcal{E})$ be a coarse space and $f:X\times X \rightarrow \mathbbm{R}_{+}$ a function. Then $f$ \textit{has controlled decay} with respect to $\mathcal{E}$ if for all $\epsilon>0$, there is a controlled set $E\in \mathcal{E}$ such that $f(x,y)<\epsilon$ for all $(x,y)\notin E$.
\end{defn}

Given a real-valued function on $X\times X$, we will now describe a "universal" coarse structure with respect to which $f$ has controlled decay.

\begin{thm} \label{univ. cs} For any function $f: X\times X \rightarrow \mathbbm{R}_{+}\cup \{\infty\}$, there is a coarse structure $\mathcal{E}_{f}$ such that for any other coarse structure $\mathcal{E}$ on $X$, $f$ has controlled decay with respect to $\mathcal{E}$ if and only if $\mathcal{E}_{f}\subseteq \mathcal{E}$.
\end{thm} 

\begin{proof}

Let $f:X\times X\rightarrow \mathbbm{R}_{+}\cup \{\infty\}$. For $\epsilon>0$, set 

$$E_{f,\epsilon}:=\{(x,y)\in X\times X\ :\ f(x,y)\ge \epsilon \}\subseteq X\times X,$$ \noindent and define the monogenic coarse structure $\mathcal{E}_{f,\epsilon}:=\langle E_{f,\epsilon}\rangle$. Then for $\epsilon\le \epsilon^{\prime}$, we have $\mathcal{E}_{\epsilon^{\prime}}\subseteq \mathcal{E}_{\epsilon}$, and thus we have an ``order reversing" poset morphism from $\mathbbm{R}_{+}$ to coarse structures on $X$.

Define the coarse structure $\mathcal{E}_{f}= \bigcup_{\epsilon>0} \mathcal{E}_{f,\epsilon}$. Note that this is a coarse structure, since the $\mathcal{E}_{f,\epsilon}$ are nested. By construction, $f$ will have controlled decay with respect to any coarse structure containing $\mathcal{E}_{f}$. 

Conversely, suppose that $f$ has controlled decay with respect to the coarse structure $\mathcal{E}$. Then for each $\epsilon>0$, there exists an $E\in \mathcal{E}$ with $E_{f,\epsilon}=\{ (x,y)\ :\ f(x,y)\ge \epsilon\}\subseteq E$. Thus $\mathcal{E}_{f}=\bigcup_{\epsilon>0} \mathcal{E}_{f,\epsilon}\subseteq \mathcal{E}$. 
\end{proof}.

\begin{rem}\label{G-epislon}The intermediate coarse structures $\mathcal{E}_{f,\epsilon}$ in the proof will prove to be useful. Indeed, for a given $\epsilon>0$, consider the graph $G_{f,\epsilon}$, whose vertices are $X$ and there is an edge between $x$ and $y$ if $(x,y)\in E_{f,\epsilon}$. Equip this graph with the path metric, and if $G_{f,\epsilon}$ is connected, it is is coarsely equivalent to $(X,\mathcal{E}_{f,\epsilon})$.
\end{rem}

We have the following proposition.

\begin{prop}\label{coarsegeodesic} The coarse structure $\mathcal{E}_{f}$ is monogenic if and only if $\mathcal{E}_{f}= \mathcal{E}_{f,\epsilon} $ for some $\epsilon>0$.
\end{prop}

\begin{proof}
    Clearly if $\mathcal{E}_{f}=\mathcal{E}_{f,\epsilon} $ then $\mathcal{E}_{f}$ is monogenic. Conversely, suppose that $\mathcal{E}_{f}$ is monogenic, with generator $E$. Then since $\mathcal{E}_{f}=\bigcup_{\epsilon} \mathcal{E}_{f,\epsilon}$, we have $E\in \langle \mathcal{E}_{f,\epsilon_{0}}\rangle$ for some $\epsilon_{0}>0$. Thus $\mathcal{E}_{f}=\langle E\rangle\subseteq \mathcal{E}_{f,\epsilon_{0}}\subseteq \mathcal{E}_{f}$, giving an equality $\mathcal{E}_{f,\epsilon_{0}}=\mathcal{E}_{f}$
\end{proof}

\begin{defn}
We call a function $f\times f\rightarrow \mathbbm{R}_{+}$  \textit{coarsely stable} if $\mathcal{E}_{f}$ is monogenic, or equivalently for some $\epsilon>0$, the bounded coarse structure of the graph $G_{f,\epsilon}$ is $\mathcal{E}_{f}$. 
\end{defn}

Note that if $\epsilon$ is stable, then $\epsilon^{\prime}$ is stable for any $\epsilon^{\prime}<\epsilon$. We call the corresponding interval the \textit{stable range}. We note that stability is a \textit{property of the coarse structure $\mathcal{E}_{f}$}.

If $f$ is stable, by choosing an $\epsilon$ we can recover an explicit metric model for $\mathcal{E}_{f}$ given by the path metric on $G_{f,\epsilon}$.  Notice that the actual metric space depends heavily on $\epsilon$ in general, but we might hope that they will all be \textit{coarsely equivalent} when $\epsilon$ is sufficiently small. For a function $f:X\times X\rightarrow \mathbbm{R}_{+}$, if $\mathcal{E}_{f}$ is monogenic, then we call the set $\{\epsilon>0\ :\ \mathcal{E}_{f}=\mathcal{E}_{f,\epsilon}\}$ the \textit{stable range} for $f$.

We note that if $\mathcal{E}_{f}$ is monogenic, then the condition of properness of the coarse structure is equivalent to $G_{f,\epsilon}$ being locally finite.

\bigskip

\noindent \textbf{Lower semi-continuity of coarse structures}. Consider the complete metric space $B_{+}(X\times X)$ of bounded, non-negative real valued functions on $X\times X$ with the sup norm as metric, i.e.

$$d(f,g):=\text{sup}_{(x,y)\in X\times X} |f(x,y)-g(x,y)|.$$

It would be very useful if the the assigment $f\mapsto \mathcal{E}_{f}$ exhibited some form of continuity, for example, by allowing us to glean information about $\mathcal{E}_{f}$ from coarse structures $\mathcal{E}_{f_{\lambda}}$, where $f_{\lambda}\rightarrow f$ is a net of functions converging to $f$. To this end, we have the following proposition, which can be interpreted as lower semi-continuity of the map $B_{+}(X\times X)\rightarrow \text{Coa}(X)$, where the latter denotes the set of coarse structures on $X$ naturally viewed as a poset with inclusion.

\begin{thm}\label{lowersemicont}
    If $\Lambda$ is a directed set and  $f_{\lambda}\rightarrow f$ is a convergent net over $\Lambda$ in $B_{+}(X\times X)$, then $\mathcal{E}_{f}\subset \langle \{\mathcal{E}_{f_{\lambda}}\}_{\lambda\in \Lambda}\rangle$.
\end{thm}

\begin{proof}

If $\|f-g\|<\delta$. Then for any $\epsilon>\delta$ $E_{f,\epsilon}\subseteq E_{g,\epsilon-\delta}$ for all $\epsilon$, theorefore $\mathcal{E}_{f,\epsilon}\subseteq \mathcal{E}_{g}$. Now, for every $\epsilon>0$, choose some $\lambda_{\epsilon}$ with $d(f,f_{\lambda_{\epsilon}})<\epsilon$. Then we see that $$\mathcal{E}_{f}\subseteq \langle \{{\mathcal{E}_{f_{\lambda_{\epsilon}}}}\}_{\epsilon}\rangle\subseteq \langle \{\mathcal{E}_{f_{\lambda}}\}_{\lambda}\rangle$$

\end{proof}

\bigskip

\noindent \textbf{Higher order controlled decay}. Commonly in the literature, decay results for correlators and commutators frequently involve not just operators localized at a pair of sites, but at multiple sites. These are called \textit{higher order correlators.} One way to think about this is that the coarse structure is ``invariant under coarse-graining". Let $A$ be a discrete net of C*-algebras over $X$. If $(X,d)$ is a metric space, then for every $R>0$ we can ``coarse-grain" the system, by replacing the algebra $A_{x}$ of a point with $A_{N_{R}(x)}$, the algebra of operators localized in an R-ball around $x$. Then $\phi$ has higher decay of correlations if for all $R>0$, the new rescaled correlation functions have controlled decay.

In the majority of cases of interest, physical states satisfy this property with respect to the metric in question. However, what could go wrong is that it is possible to have states which are local with respect to the metric, but after coarse graining we  obtain correlations that appear non-local. This leads us to define a new coarse structure which is universal with respect to higher order decay of correlations. As in the previous section, we abstract away the physics and treat this at the level of functions.

\begin{defn} \label{higher cont dec}Suppose $\widetilde{f}:\mathcal{P}(X)\times \mathcal{P}(X)\rightarrow \mathbbm{R}_{+}$ is a function, and $\mathcal{E}$ is a coarse structure on $X$. Then $\widetilde{f}$ has \textit{higher controlled decay} if for every $E\in \mathcal{E}$, the function $\widetilde{f}_{E}:X\times X\rightarrow \mathbbm{R}_{+}$ defined by $\widetilde{f}_{E}(x,y):=\widetilde{f}(E[x],E[y])$ has controlled decay.
\end{defn}

\begin{thm}\label{scaledef} There exists a coarse structure, $\mathcal{E}_{\widetilde{f}}$, such that for any coarse structure on $\mathcal{E}$ such that $\widetilde{f}$ has higher controlled decay, $\mathcal{E}_{\widetilde{f}}\subseteq \mathcal{E}$.
\end{thm}

\begin{proof}

Letting $D$ be the diagonal of $X$, define the function $f:=\widetilde{f}_{D}:X\times X\rightarrow \mathbbm{R}_{+}$. Then set $\mathcal{E}_{0}=\mathcal{E}_{f}$, and $$\mathcal{E}_{1}:=\langle \bigcup_{E\in \mathcal{E}_{0}}\mathcal{E}_{\widetilde{f}_{E}}\rangle.$$

inductively define

$$\mathcal{E}_{n+1}:=\langle \bigcup_{E\in \mathcal{E}_{n}}\mathcal{E}_{\widetilde{f}_{E}}\rangle.$$

Then we have an increasing sequence of coarse structures $\mathcal{E}_{f}=\mathcal{E}_{0}\subseteq\mathcal{E}_{1}\subseteq\dots$ and we set $$\mathcal{E}_{\widetilde{f}}=\bigcup_{n} \mathcal{E}_{n}$$

First we note that $\widetilde{f}$ has higher controlled decay with respect to $\mathcal{E}_{\widetilde{f}}$. Indeed, let $E\in \mathcal{E}_{\widetilde{f}} $. Then there exists some $n$ such that $E\in \mathcal{E}_{n}$. By construction $\mathcal{E}_{\widetilde{f}_{E}}\in \mathcal{E}_{n+1}\subseteq \mathcal{E}_{\widetilde{f}}$, hence $\widetilde{f}_{E}$ has controlled decay with respect to $\mathcal{E}_{\widetilde{f}}$.

Now, let $\mathcal{E}$ be any coarse structure on $X$ for which $\widetilde{f}$ has higher order controlled decay. Then using the diagonal controlled set in $\mathcal{E}_{\widetilde{f}}$, we see that $f$ has controlled decay with respect to $\mathcal{E}_{\widetilde{f}}$ and thus $\mathcal{E}_{0}=\mathcal{E}_{f}\subseteq \mathcal{E}$. Now, for any $E\in \mathcal{E}_{0}$, $\widetilde{f}_{E}$ has controlled decay so $\mathcal{E}_{1}\subseteq \mathcal{E}$. Continuing inductively we see $\mathcal{E}_{\widetilde{f}}\subseteq \mathcal{E}$ as desired.
\end{proof}

\subsection{Asymptotic dimension and growth rate}\label{asymptoticdimgrowthrate}
Our definition of a coarse structure from a function $f:X\times X\rightarrow \mathbbm{R}_{+}$ works well in theory, but there are obstacles to trying to use this perspective in practice. For one, a coarse structure can be a very complicated piece of data to completely specify. Another difficulty is that our formalism is based on the philosophy of \textit{asymptotic reasoning} \cite{MR2038580}, commonplace in condensed matter physics. The coarse structures we have formally defined are mathematical idealizations, and are not literally describing exact features of the physical systems we are modeling, since these will typically have only finitely many degrees of freedom while our coarse structures are witnessed only when there are infinitely many degrees of freedom. This type of idealization is necessary to rigorously define concepts such as phase transitions.  

Thus in our case we have implicitly assumed the set $X$ is infinite, as in the thermodynamic limit of a spin system, in order to obtain non-trivial coarse structures. However, real-world systems to which we would like to apply our theory (e.g. states prepared on a quantum processor) will be finite dimensional, and in particular will have finite number of sites $X$. This begs the question, in what sense can we try to approximate the infinite limit coarse structure in a finite system, in a way that might be accessible empirically? 

One way to approach both of these issues is to find \textit{invariants} of coarse spaces that can, in some sense, be approximated on a large (but finite) piece. The invariant we will focus on here is the \textit{asymptotic dimension}. There are multiple definitions for the asymptotic dimension. Here we use the one due to \cite{MR4646531} which emphasizes the case of metric spaces. For the more general version, we refer the reader to \cite{MR2007488}.

First we introduce some notation. Suppose $\mathscr{U}:=\{\mathit{U}_i\}_{i\in I}$ is a cover for a metric space $M$. For $R>0$, we define the $R-$multiplicity of $\mathscr{U}$ as the smallest integer $n$ such that for every $x\in M$ the ball $B_{R}(x)$ intersects at most $n$ elements of $\mathscr{U}$.

\begin{defn}\label{asymptoticdim}\cite{MR4646531}
    Let $M$ be a metric space. The \textit{asymptotic dimension} of $M$ is be the smallest integer $n\ge 0$ such that for every $R>0$, there is a uniformly bounded cover $\mathscr{U}$ with $R$-multiplicity $n+1$. We denote the asymptotic dimension by asdim $M$.
\end{defn}

The asymptotic dimension is meant to capture the dimensionality of a space from a large-scale perspective. In particular, it is an invariant of coarse spaces \cite{MR2007488}. We have $\text{asdim}(\mathbbm{Z}^{n})=n$ for any of the usual metrics \cite{MR4646531}. This number allows us to access the \textit{coarse dimensionality} of the state.
However, the above definition suggests this number might be hard to compute in practice, and it is not clear that one could get a good estimate for this on a large, finite subset. However, using the idea of growth functions, we can attempt to approach this issue.

First we consider again the case of discrete metric space $(X,d)$ which has bounded geometry, i.e., for each $r>0$, $\sup_{x\in X}|B_{r}(x)|<\infty$. Then we define the \textit{growth function} of $(X,d)$ to be
    
    $$\gamma(r)=\sup_{x\in X}|B_{r}(x)|.$$

For example, for any metric on $\mathbbm{Z}^{n}$ such that the identity is a coarse equivalence with $\mathbbm{R}^{n}$ , we have $\gamma(r)\le Cr^{n}$, and hence these spaces are said to have \textit{polynomial growth}. For an infinite $n$-regular tree (sometimes called a Cayley tree), we have $\gamma(r)\le n(n-1)^{r-1}$, with equality for $r\in \mathbbm{N}$, which is super-polynomial. However, based on the first example, we might expect that, for spaces with polynomial growth, the asymptotic dimension should be extractable from the smallest such exponent. 

The following result gives us something in this direction:

\begin{thm}(for example, see \cite{MR4619565})
    Let $(X,d)$ be a bounded geometry metric space arising from the path metric of a connected graph. If $\lim_{r\to \infty} \frac{\gamma(r)}{r^{k+1}}=0$, then asdim $(X,d)\le k$
\end{thm}

Using this result, we can attempt to estimate the asymptotic dimension of the coarse space in the thermodynamic limit on a large but finite system by finding the smallest $k$ making this limit work ``approximately". Suppose now that $X$ is finite, and we have empirically determined the function $f:X\times X\rightarrow \mathbbm{R}_{+}$ via some measurements (perhaps $f$ is simply the correlations between observables localized at some sites). Pick some $\epsilon>0$, and construct the graph $G_{f,\epsilon}$ where the vertices are $X$ and $(x,y)$ is an edge if $f(x,y)> \epsilon$, as in the previous section. Then calculate the growth function $\gamma_{\epsilon}(r)$ for this graph and look at $$\frac{\gamma_{\epsilon}(r)}{r^{k+1}}.$$

\noindent In the thermodynamic limit, we would be looking for which values of $k$ make this ratio go to $0$ as $r\rightarrow \infty$, and then we would try to find the smallest such $k$. There are obvious problems with this approach taken literally in the finite setting. For example, if the smallest value of $f$ is non-zero, once $\epsilon$ gets below this value, $G_{f,\epsilon}$ degenerates to a complete graph. Similarly, since each graph is finite, as $r\rightarrow \infty$ the function $\gamma_{\epsilon}(r)$ becomes constant after $r$ reaches the diameter of $G_{f,\epsilon,}$, hence $\frac{\gamma_{\epsilon}(r)}{r^{k+1}}\rightarrow 0$ for \textit{any k}. Thus it will be neccessary to look in a ``middle range" of $\epsilon$ and $r$ to see the most persistant features of these quantities, directly analogous to the methodology of persistant homology used in topological data analysis. We hope to pursue this practical perspective in future work.

In a slighlty different direction, we can also use the idea of growth functions to help us distinguish coarse spaces in a way that will be useful to us when comparing coarse structures. There is a general theory of growth functions for bounded geometry coarse structures \cite{MR2007488}, but we will restrict ourselves, for the purpose of simplicity, to uniformaly discrete, bounded geometry metric spaces. In this paper, we will frequently be in the situation where we have a coarse structure $(X,\mathcal{E})$ and are able to conclude some other coarse structure $\mathcal{E}^{\prime}$ on the same set $X$ such that $\mathcal{E}^{\prime}\subseteq \mathcal{E}$. In general, it is not the case that asymptotic dimension is decreasing with respect to inclusions of coarse structures in this sense. The question is, then, how much information does the fact of this inclusion give us? In particular, can we use this to rule out a large class of possibilities for the coarse equivalence class of $(X,\mathcal{E}^{\prime})$ if we know the coarse equivalence class of $(X,\mathcal{E})$, and vice-versa?

We have the following proposition in this direction:

\begin{prop} Let $X$ and $Y$ be uniformly discrete bounded geometry metric spaces, and suppose there is a function $f:\mathbbm{R}_{+}\rightarrow \mathbbm{R}_{+}$ with $\gamma_{X}(r)\le f(r)$ and 
$$\text{limsup}_{r}\frac{\gamma_{Y}(r)}{f(r)}=\infty.$$

Then $Y$ is not coarsely equivalent to $(X,\mathcal{E})$ for any coarse structure $\mathcal{E}$ on $X$ contained in the metric coarse structure.
\end{prop}

\begin{proof}
Let $\mathcal{F}$ denote the metric coarse structure on $Y$. We will prove the the contrapositive of the claim. 

Suppose $(Y,d_{Y})$ is coarsely equivalent to $(X,\mathcal{E})$ for some coarse structure $\mathcal{E}$ contained in the metric coarse structure of $X$. Then we have functions $\rho:Y\rightarrow X$ and $\pi:X\rightarrow Y$ such that $\rho\times \rho(\mathcal{F})\subseteq \mathcal{E}$, $\pi\times \pi(\mathcal{E})\subseteq \mathcal{F}$, and the sets $\{(y,\pi\circ \rho(y))\ :\ y\in Y\}\in \mathcal{F}$ and $\{(x,\rho\circ \pi(y))\ :\ x\in X\}\in \mathcal{E}$.

First our claim is that $\sup_{x\in X} | \rho^{-1}(x)|<M$ for some finite $M$. Suppose for contradiction this is not the case. Since $(Y,d_{Y})$ is a uniformly discrete, bounded geometry metric space, $\gamma(n)\rightarrow \infty$. Then there exists a sequence $\{x_{n}\}\subseteq X$ such that $|\rho^{-1}(x_{n})|\ge \gamma_{Y}(n)$. This implies for each $n$, there exists $y_{n},z_{n}\in \rho^{-1}(x_{n})$ with $d_{Y}(y_{n},z_{n})\ge n$.

But by assumption $F:=\{(y_n, \pi\circ \rho(y_{n}))\}, \ G:=\{(\pi\circ \rho(z_n), z_{n})\}\in \mathcal{F}$. But $\pi\circ \rho(y_{n}))=\pi(x_{n})=\pi\circ \rho(z_{n}))$, hence

$$\{(y_{n},z_{n})\}\subseteq F\circ G\in \mathcal{F}$$

But since $\mathcal{F}$ is the bounded coarse structure for $(Y,d_{Y})$, there exists $N$ such that $d(y_{n},z_{n})\le N$ for all $n$, a contradiction.

Now, let $f$ be a function as in the hypothesis of the theorem. Then we have for any $y\in Y$ and $r> 0$, we have $|B_{r}(y)|\le M|B_{r}(\rho(y))|\le M f(r)$, hence  

$$\frac{\gamma_{Y}(r)}{f(r)}\le M$$

\begin{ex}\label{subsequivalenceexample} As an application of the above theorem, consider the metric spaces $\mathbbm{Z}^{n}$ with the (coarse equivalence class of) the Euclidean metric. Then the literal size of balls $\mathbbm{Z}^{n}$ will depend on the details of the metric you choose, but will asymptotically behave like $Cr^{n}$ for some constant $C$. Thus for any $m>n$, the above theorem implies we cannot have  $\mathbbm{Z}^{m}$ (with the usual metric course structure) equivalent to a subcoarse structure of $\mathbbm{Z}^{n}$ (with the usual metric coarse structure). 
\end{ex}

\end{proof}

\subsection{Correlation coarse structure}

Let $A$ be a discrete net of C*-algebras over $F$, let $B$ be any Banach algebra, and denote by $\mathcal{B}_{1}(A\rightarrow B)$ the space of linear maps from $A$ to $B$ with norm one. Then for $\phi\in \mathcal{B}_{1}(A\rightarrow B)$ we define the 2-point correlation functions

$$C_{\phi}(x,y)=\sup_{a\in (A_{x})_{1},b\in (A_{y})_{1}} \|\phi(ab)-\phi(a)\phi(b)\|.$$

and the higher order correlation

$$\widetilde{C}_{\phi}(F,G)=\sup_{a\in (A_{F})_{1},b\in (A_{G})_{1}} \|\phi(ab)-\phi(a)\phi(b)\|.$$

 We will be interested in the case where $B$ is a C*-algebra and $\phi$ is a ucp map (quantum channel), and even more specifically when $B=\mathbbm{C}$ so that $\phi$ is a state. In this case we can interpret $C_{\phi}(x,y)$ as the maximal correlation of observables localized at the sites $x$ and $y$ respectively. This is the main case we will be interested in.

\begin{defn} We define the correlation coarse structures by $\mathcal{E}_{\phi}:=\mathcal{E}_{C_{\phi}}$ and $\mathcal{E}_{\widetilde{\phi}}:=\mathcal{E}_{\widetilde{C}_{\phi}}$
\end{defn}

We have the following inclusions of coarse structures 
$$
\begin{tikzcd} 
\mathcal{E}_{\phi}\arrow[r, phantom, sloped, "\subseteq"]\arrow[d, phantom, sloped, "\subseteq"] & \mathcal{E}_{\widetilde{\phi}}\arrow[d, phantom, sloped, "\subseteq"]\\
(\mathcal{E}_{\phi})_{\text{con}}\arrow[r, phantom, sloped, "\subseteq"] & (\mathcal{E}_{\widetilde{\phi}})_{\text{con}}\\
\end{tikzcd}
$$

\noindent where the bottom row are the connected quotients as in Example \ref{connected quotient}.

 The motivation is that many states associated with a Hamiltonian, local with respect to some metric $(X,d)$, have (higher order) decay of correlations. We use the term ``local" for Hamiltonian loosely, but typically it means the interaction terms become smaller in norm as the support sets grow in diameter, usually with some specified bound. Similarly we will say that a state is \textit{local} with respect to some coarse structure to mean the coarse structure conatins either $\mathcal{E}_{\phi}$ or $\mathcal{E}_{\widetilde{\phi}}$. We list some of the motivating examples below.
    
\begin{ex}\textbf{Ground states of gapped Hamiltonains}
For a local Hamiltonian of a quantum spin system on a discrete metric space, if $H$ has a unique ground state then this state will generically have correlations with (higher order) controlled decay with respect to the metric \cite{Hastings:2005pr, MR2217299}. The function $C_{\phi}$ will generically have exponential decay. A subclass of gapped Hamiltonains are those arising from commuting projector Hamiltonians with local topological order \cite{10.1063/1.3490195}. The correlations of these ground states will typically have controlled support.
\end{ex}

\begin{ex} \textbf{Equilibrium states}.
In classical or quantum spin systems with a local Hamiltonian on a discrete metric space, ``thermal equilibrium" states (Gibbs states or KMS states) have controlled (typically exponential) decay at low and high temperatures {Araki, Gibbs states of a one dimensional quantum lattice}. Near critical temperatures, however, the correlations of equilibrium states tend to decay super exponentially. The decay class is usually algebraic, and is indicative of scale-free behavior \cite{Bluhm2022exponentialdecayof}, \cite{10.1063/1.4921305}, \cite{PhysRevX.4.031019}

\end{ex}

\begin{ex}\label{prepared} \textbf{States prepared by locality preserving channels}
In the previous examples, the states in question are all associated to a local Hamiltonian. However, it is fairly easy to construct states local with respect to a coarse structure without reference to any Hamiltonian. Let $A$ be a discrete net over $X$ and let $\mathcal{E}$ be a fixed coarse structure. If $\Psi:A\rightarrow A $ is a locality preserving channel with respect to $\mathcal{E}$ (See Definion \ref{locality pres}), then in [Corollary \ref{dynamics and states}], we show that for any product state $\phi$, $\mathcal{E}_{\widetilde{\phi\circ \Psi}}\subseteq \mathcal{E}$. Locality preserving channels are easy to construct (e.g. (quantum) cellular automata or even more concretely, finite depth circuits, see Example \ref{QCA}), and thus we can easily construct a large variety of states whose coarse structures are (at least contained in) any desired coarse structure without having to rely on any a Hamiltonian to mediate locality. 
\end{ex}

\begin{ex} \textbf{Spin systems on Cayley graphs}. Let $\Gamma$ be a finitely generated discrete group. Recall that the set $\Gamma$ has a canonical coarse structure arising from a \textit{Cayley graph} for the group. Pick a finite generating set $S$, and connect points $x$ and $y$ by an edge if there is a generator $g\in S$ such that $y=gx$. The coarse structure for the path metric on this graph is indepenedent of the choice of finite generating set, and we denote it $\mathcal{E}_{\Gamma}$. We will assume we have picked some $S$ once and for all, and denote the metric by $d_{\Gamma}$.

We note that in this setting, there is a right action of $\Gamma$ on $\Gamma$, defined by $R_{g}(h):=hg$. Note that $R_{g}$ are isometries of $(\Gamma,d_{\Gamma})$

Now let $A$ be a discrete net of C*-algebras over $\Gamma$, where the latter is viewed as a set. An action $\Gamma\rightarrow \text{Aut}(A)$ is called \textit{covariant} if $g(A_{F})\subseteq A_{R_{g^{-1}}(F)}$. The standard example is a generalized spin system $A:=\otimes_{x\in \Gamma} A_{x}$, where all the $A_{x}$ are equal, and $\Gamma$ acts by permuting the tensor factors on the right.

The following theorem is a tool which makes it possible to build many examples of states local with respect to the coarse structure $\mathcal{E}_{\Gamma}$ by considering a state for a spin system on a Cayley graph invariant under translation symmetry.

\end{ex}

\begin{thm} Let $\Gamma$ be a finitely generated discrete group and $A$ a locally finite dimensional discrete covariant net of C*-algebras over $\Gamma$. For any $\Gamma$-invariant factor state $\phi$ on $A$, $\mathcal{E}_{\widetilde{\phi}}\subseteq \mathcal{E}_{\Gamma}$. In particular, if $\phi$ is the unique ground state for a translation invariant Hamiltonian on $A$, then $\mathcal{E}_{\widetilde{\phi}}\subseteq \mathcal{E}_{\Gamma}$.
\end{thm} 

\begin{proof}
For $R\ge 0$, let $E_{R}:=\{(x,y)\in \Gamma\times \Gamma\ :\ d_{\Gamma}(x,y)\le R\}$. It suffices to check $(C_{\phi})_{E_{R}}$ has controlled decay for all $R\ge 0$.
Let $R, \epsilon>0$. Let $e\in \Gamma$ be the identity, and choose a basis $\{a_{1},\dots, a_{n}\}$ of $A_{N_{R}(e)}$ such that $||a_{i}||\le 1$. Since all norms on a finite-dimensional C*-algebra are equivalent, there is some constant $D>0$ such that if $a=\sum \lambda_{i}a_{i}$, $\sum_{i} |\lambda_{i}|\le D\|a\|$.

Now, recall a net of C*-algebras in our sense gives us a \textit{quasi-local} C*-algebra in the sense of Bratteli and Robinson \cite{MR887100} where the poset in question is the set $\mathcal{P}(\Gamma)$ of all subsets of $\Gamma$, and the $\perp$ relation is disjointness. Then by \cite[2.6.5]{MR887100}, since $\phi$ is a factor state, there exists finite subsets $F_{1}, \dots, F_{n}$ such that $\sup_{b\in (A_{F^{c}_{i}})_{1}}|\phi(a_{i}b)-\phi(a_{i})\phi(b)|<\frac{\epsilon}{D}$.

Choose $S>0$ sufficiently large such that $F_{1},\dots, F_{n}, N_{R}(e)\subseteq N_{S}(e)$. Then for any $a=\sum \lambda_{i}a_{i}\in (A_{N_{R}(e)})_{1}$ and $b\in A_{N_{S}(e)^{c}}$, we have 

$$|\phi(ab)-\phi(a)\phi(b)|\le \left(\sum|\lambda_{i}|\right)\frac{\epsilon}{D}<\epsilon.$$

Hence if we choose $T>R+S$, we have for any $x\in \Gamma$ with $d_{\Gamma}(e,x)>T$, 

$$(C_{\phi})_{E_{R}}(e,x)<\epsilon$$

Now, for any $y,z\in \Gamma$ with $d_{\Gamma}(y,z)>T$, then $d_{\Gamma}(e,zy^{-1})>T$, and by translation invariance of $\phi$ we have

$$(C_{\phi})_{E_{R}}(y,z)=(C_{\phi})_{E_{R}}(e,zy^{-1})<\epsilon.$$

For the final statement, if a ground state of a Hamiltonian $H$ is unique, then it is pure, hence a factor state \cite{MR887100}, and if $H$ is translation invariant, so is its unique ground state.
\end{proof}

It is easy to show that a if a discrete group $\Gamma$ is \textit{amenable}, then one can always do an avergaing type argument to actually produce an invariant state.

\begin{rem} In all the above examples, we have results that guarantee that $\phi$ is local with respect to some coarse structure $\mathcal{F}$. However, it is fairly easy to come across sufficient conditions that yields the stronger statement $\mathcal{F}=\mathcal{E}_{\phi}$. For example, if the state has some non-trivial uniform degree of entanglement between nearby sites $x,y$, then this should suffice. More precisely, suppose we are working on with coarse structure $\mathcal{F}$ corresponding to the quasi-geodesic metric space $(X,d)$, and that there exists $R>0$ and $c>0$ such that if $d(x,y)\le R$ then $C_{\phi}(x,y)\ge c$. Then we have $\mathcal{F}\subseteq \mathcal{E}_{\phi}$. This condition will occur fairly generically, for example, in states constructed from finite depth quantum circuits as in Example \ref{prepared}.
\end{rem}

\bigskip

\noindent \textbf{Close states}. Note that the space of bounded linear maps $\mathcal{B}_{1}(A\rightarrow B)$ with norm $1$ has a natural topology, with metric given by the norm $$||\phi-\psi||=\sup_{a\in (A)_{1}} \|\phi(a)-\psi(a)\|.$$

First we will show that the function $\mathcal{B}_{1}(A\rightarrow B)\rightarrow \text{BdFun}(X\times X)$ is continuous, where the latter denotes the space of bounded real-valued functions on $X\times X$ with the sup norm $|| \cdot ||_{\infty}$, which combined with Theorem \ref{lowersemicont} will give a lower semicontinuity of the map $\phi\mapsto \mathcal{E}_{\phi}$.

\begin{prop}\label{stability1}
Let $\phi,\psi\in \mathcal{B}_{1}(A\rightarrow B)$ such that $\|\phi-\psi\| \leq \delta$. Then $\|C_{\phi}-C_{\psi}\|_{\infty} \leq 3 \delta $. As a consequence, 

\begin{enumerate}
\item 
$\phi\mapsto C_{\phi}$ is continuous and $\phi\mapsto \mathcal{E}_{\phi}$ is lower semicontinuous in the sense of Theorem \ref{lowersemicont}.
\item 
If $\epsilon-3\delta>0$, then $\mathcal{E}_{\epsilon+3\delta,\phi}\subseteq\mathcal{E}_{\epsilon,\psi}\subseteq \mathcal{E}_{\epsilon-3\delta,\phi}$.
\end{enumerate}
\end{prop}

\begin{proof}

For any $a\in (A_{x})_{1}, b\in (A_{y})_{1}$, we have

\begin{align*}
\| \phi(ab)-\phi(a)\phi(b)\|&\le \|\phi(ab)-\psi(ab)\|+\|\phi(a)\phi(b)-\psi(a)\psi(b)\|+\|\psi(ab)-\psi(a)\psi(b)\|\\
&\le 3\delta + \|\psi(ab)-\psi(a)\psi(b)\|
\end{align*}

Taking the supremum over $a\in (A_{x})_{1}, b\in (A_{y})_{1} $, we obtain

$$C_{\phi}(x,y)\le C_{\psi}(x,y)+3\delta$$

By switching $\phi$ and $\psi$ and putting these together, we obtain

$$||C_{\phi}-C_{\psi}||\le 3\delta.$$

If $(x,y)\in E^{\phi}_{\epsilon}=\{(x,y)\in X\times X\ : C_{\phi}(x,y)\ge \epsilon\}$, then $C_{\phi}(x,y)\ge C_{\phi}(x,y)-3\delta$, and thus $(x,y)\in E^{\psi}_{\epsilon-3\delta}$. Thus 

$$\mathcal{E}_{\epsilon, \phi}=\langle E^{\phi}_{\epsilon}\rangle\subseteq \langle E^{\psi}_{\epsilon-3\delta} \rangle=\mathcal{E}_{\epsilon-3\delta, \psi}.$$

The other inclusion is obtained by switching $\phi$ and $\psi$.

\end{proof}

\begin{cor}
    Let $\phi, \psi\in \mathcal{B}_{1}(A\rightarrow B)$ such that $\mathcal{E}_{\phi}$ and $\mathcal{E}_{\psi}$ are monogenic, and let $L$ be the length of the largest (open) interval contained in the itnersection of the stable ranges of $\phi$ and $\psi$. If  $\|\phi-\psi\|<\frac{L}{6}$, then $\mathcal{E}_{\phi}=\mathcal{E}_{\psi}$.
\end{cor} 

\begin{proof}
Let $\delta=\|\phi-\psi\|$ and choose $\epsilon$ in the intersection of the two stable ranges such that $(\epsilon-3\delta, \epsilon+3\delta)\subseteq (0,L)$. Then from the above proposition, we see

$$\mathcal{E}_{\phi}=\mathcal{E}_{\epsilon+3\delta, \phi}\subseteq \mathcal{E}_{\epsilon, \psi}\subseteq \mathcal{E}_{\epsilon-3\delta, \phi}=\mathcal{E}_{\phi},$$

hence $\mathcal{E}_{\phi}=\mathcal{E}_{\epsilon,\psi}=\mathcal{E}_{\psi}$.
\end{proof}

We would like it to be the case that if states are sufficiently close, their coarse structures agree on the nose. However, as we have pointed out the above proposition does not give us enough information to accurately compare coarse structures. However, with some additional assumptions we will be able to conclude such a result. For any state $\phi$, let $\phi^{u}$ denote the state $\phi(u\cdot u^{*})$.

\begin{rem}\label{Tensor prod 1} \textbf{Choice of tensor product decomposition for states on spin systems}. We emphasize here that the correlation coarse structure depends heavily on our choice of local subalgebras. For example, we can consider different tensor product decompositions of spin systems. In concrete contexts (e.g. simulating a state on a quantum processor) it is easy to make a natural choice for the tensor product decomposition. However, it is interesting to ask how different things can look with respect to a different tensor product decomposition. Suppose $\phi,\psi\in \mathcal{B}_{1}(A\rightarrow B)$ , and we have an automorphism $\alpha: A\rightarrow A$ such that $\phi\circ \alpha=\psi$. Then consider a new point-structure defined over the same base set $X$ by $\{A^{\prime}_{x}=\alpha(A_{x})\}_{x\in X}$. Note that this site structure is local (resp. generating) if and only if $\{A_{x}\}_{x\in X}$ is. Now, with the new site structure, we compute

\begin{align*}C^{\prime}_{\phi}(x,y)&=\sup_{a^{\prime}\in (A^{\prime}_{x})_{1},b^{\prime}\in (A^{\prime}_{y})_{1}} \|\phi(a^{\prime}b^{\prime})-\phi(a^{\prime})\phi(b^{\prime})\|\\
&=\sup_{a\in (A_{x})_{1},b\in (A_{y})_{1}} \|\phi(\alpha(ab))-\phi(\alpha(a))\phi(\alpha(b))\|\\
&=C_{\psi}(x,y)
\end{align*}

Thus with respect to the new site structure $\{A^{\prime}_{x}\}$, the two point correlation function for $\phi$ looks exactly like the correlation function for $\psi$ with respect to the original coarse structure.

To reiterate, two different site structures $\{A^{\prime}_{x}\}_{x\in X}$ and $\{A_{x}\}_{x\in X}$ for the same state $\phi$ can give wildly different coarse structures and decay rates. In particular, we could have a state that from one site structure looks like a 3D state with 0 correlation length, while with respect to another site structure, it looks like a 2D  state with infinite correlation length (for formal definitions of correlation length, see \ref{defn:correlationlength})

The obvious explanation for what is going on is that the automorphism $\alpha$ taking $\phi$ to $\psi$ is wildly ``non-local". In other words, we expect the point-like operators $a\in A^{\prime}_{x}$ to be wildly non-local with respect to the original site structure. We might imagine that if we assume some sort of ``locality-preserving" type of condition  on $\alpha$ that maybe we would fare better in this regard. See section \ref{loc pres chan} for further discussion.
\end{rem}

\subsection{Dynamical coarse structure}\label{subsection 4.3}

In our definition of discrete nets of algebra, we have a built-in algebraic version of locality, which is a formalization of the intuition that for observables localized in different regions of space, for \textit{any} state of the system we expect to be able to form a joint probability distribution for the measurement outcomes (i.e. the elements of the algebra commute). This shows that for quantum systems, we expect spatial separation to be connected to commutativity of localized operators.

An important dynamical manifestation of this idea is the \textit{Lieb-Robinson bound} for quantum spin systems \cite{MR312860, PhysRevB.69.104431, MR2217299}. Suppose $H$ is a (for example, strictly) local Hamiltonian for a spin system on a square lattice. Let $\alpha_{t}$ be the automorphism of the quasi-local algebra obtained by evolving the system for time $t$. Then there exists positive constants $C_{t},A$ such that for $a\in A_{x}, b\in A_{y}$

$$||[\alpha_{t}(a),b]||\le C_{t} e^{-A d(x,y)}$$

This tells us that the ``spread of information" has a speed limit in quantum spin systems, and has a wide range of important theoretical consequences in mathematical physics. If we view a channel $\alpha: A\rightarrow A$ as a time evolution, this leads us to define the following:

\begin{defn} Let $A$ be a discrete net of algebras and $\alpha\in B_{1}(A\rightarrow A)$. Define $Q_{\alpha}(x,y):=\sup_{a\in (A_{x})_{1}, b\in (A_{y})_{1}} ||[\alpha(a),b]||$. Then the \textit{dynamical coarse structure} of $\alpha$ is defined as $\mathcal{E}_{\alpha}:=\mathcal{E}_{Q_\alpha}$.
\end{defn}

We note that we are often interested in unitary dynamics, where our channel $\alpha\in \text{Aut}(A)$. In this case, we have $Q_{\alpha}(x,y)=Q_{\alpha^{-1}}(x,y)$, hence $\mathcal{E}_{\alpha}=\mathcal{E}_{\alpha^{-1}}$.

\begin{ex} \textbf{Local time evolution on a spin system}.
Let $(X,d)$ be a discrete metric space, and suppose for each $x\in X$, $A_{x}\cong M_{n_{x}}(\mathbbm{C})$ for some positive integer $n_{x}$. Define $A=\otimes_{x\in X} A_{x}$. Now suppose we have a local time evolution $t\mapsto \alpha_{t}$, which is derived by integerating a local Hamiltonian (for example, satisfying the hypothesis of \cite{MR2256615}). Then the Lieb-Robinson bounds (e.g. \cite[Theorem 2.1]{MR2256615}) show that for any time $t\ge 0$, we have the inclusion of coarse structures $\mathcal{E}_{\alpha_{t}}\subseteq \mathcal{E}_{(X,d)}$. In other words, $Q_{\alpha_{t}}$ decays asymptotically with respect to the metric $d$.
\end{ex}

\begin{ex} \textbf{Local time evolution on a spin sub-system}. Suppose $L$ is a discrete metric space and let $N\subseteq L$ is a metric subspace. Pick a dimension $d$ and consider the spin sysetms with quasi-local algebras $\displaystyle B=\otimes_{L} M_{d}(\mathbbm{C})$ and $A=\otimes_{N} M_{d}(\mathbbm{C})$. Then we have a natural inclusion of quasi-local algebras $A\le B$. Furthermore, there is a canonical condition expectation $E: B\rightarrow A $ tracing out operators localized away from $L$. 

Suppose now we have a dynamics given by a quantum channel $\alpha: B\rightarrow B$, and consider the induced quantum channel $\phi:A\rightarrow A$ defined by $\phi(a):=E(\alpha(a\otimes 1))$. Then if $a\in \mathcal{A}_{x}$ and $b\in \mathcal{A}_{y}$ we compute

\begin{align*}
||[\phi(a),b]||& =||E(\alpha(a\otimes 1))b-bE(\alpha(a\otimes 1)||\\
&=||E(\alpha(a\otimes 1)(b\otimes 1)-(b\otimes 1)(a\otimes 1)||\\
&\le ||\alpha(a\otimes 1)(b\otimes 1)-(b\otimes 1)(a\otimes 1)||\\
&\le Q_{\alpha}(x,y)
\end{align*}.

Thus $Q_{\phi}(x,y)\le Q_{\alpha}(x,y)$ for all $x,y\in N$. Therefore $$\mathcal{E}_{\phi}\subseteq \mathcal{E}_{\alpha}|_{N}.$$
\end{ex}
In other words, the coarse structure is contained in the restriction to $N$ of the coarse structure on $L$ induced by  $\alpha$.

\begin{ex}\textbf{Quantum Cellular Automata}\label{QCA} . Let $A:=\otimes_{x\in X} M_{n}(\mathbbm{C})$ be a spin system and suppose $(X,d)$ is (for simplicity) a metric arising from a locally finite graph. A \textit{quantum cellular automata} (QCA) is an automorphism $\alpha\in \text{Aut}(A)$ such that $Q_{\alpha}$ has $0$ length (in the sense of Definition \ref{defn:correlationlength}) with respect to the metric coarse structure $(X,d)$. This is equivalent to the automorphism having \textit{bounded spread} \cite{https://doi.org/10.48550/arxiv.quant-ph/0405174, QCAReview1, Farrelly2020reviewofquantum}. 

For sufficiently non-trivial QCA $\alpha\in \text{Aut}(A)$, $\mathcal{E}_{\alpha}$ will recover the metric coarse structure. Suppose there exists an $a\in M_{n}(\mathbbm{C})$ such that the support of the operator $\alpha(a_{x})$ contains the ball of radius 1 about $x$, $N_{1}(x)$. Here $a_{x}$ is the identification of $a$ in the factor $M_{n}(\mathbbm{C})$ identified with the site $x$. Then for any $y$ adjacent to $x$ $[\alpha(a_{x}), A_{y}]\ne 0$, so that $Q_{\alpha}(x,y)>0$, if we set $C_{x}=\text{min}\{Q_{\alpha}(x,y)\ : d(x,y)=1\}$. Then $C_{x}>0$. If we assume sufficient regularity (e.g. some kind of translation symmetry) then $\text{inf}_{x\in X} \{C_{x}\}=C>0$. This implies that the controlled set $F=\{(x,y)\ : Q_{\alpha}(x,y)\ge C\}\in \mathcal{E}_{\alpha}$ has the adjacency relation $E:=\{(x,y)\ :\ d(x,y)=1\}$ as a subset, hence $E$ is controlled in $\mathcal{E}_{\alpha}$. But the coarse structure generated by $E$ is precisely the coarse structure of the graph, hence $\mathcal{E}_{d}\subseteq \mathcal{E}_{\alpha}$. The other inclusion follows from the definition of QCA.

One class of QCA which will provide a large class of easily constructible dynamics with accessible coarse structures are \textit{finite-depth circuits}. Suppose we have a spin system $A=\otimes_{x\in X} M_{d}(\mathbbm{C})$, and suppose we have a metric structure on $X$ arising from a discrete metric space. To define a depth one circuit, choose a partition of $X$ into finite subsets $F_{i}$ with $\sup_{i} \text{diam}(F_{i})<R$ for some constant $R>0$. Then choose a unitary $u_{i}\in A_{F_{i}}$. Note that since the $F_{i}$ are mutually disjoint, the $u_{i}$ mutually commute in $A$, so the map $\prod_{i}\text{Ad}(u_{i})$ is well defined on the algebra $A_{loc}$, and extends to an automorphism $\alpha$ of $A$ called a \textit{depth one circuit}.

A finite-depth circuit is an automorphism $\alpha\in\text{Aut}(A)$ which can be written as a composition $\alpha_{1}\circ \alpha_{2}\circ \dots \alpha_{n}$ where each $\alpha_{i}$ is a depth one circuit.

\end{ex}

\bigskip

\noindent \textbf{Close automorphisms}. There are several natural topologies on the space $\mathcal{B}_{1}(A\rightarrow A)$, but the one we consider here is the norm topology, where 

$$\| \alpha\|=\sup_{a\in A} \frac{\|\alpha(a)\|}{\|a\|}.$$

\noindent We give here an analogue of the Proposition \ref{stability1} for dynamical coarse structures.

\begin{prop}\label{dynamicalclos}
Let $\alpha,\beta\in \mathcal{B}_{1}(A\rightarrow A)$ such that $\|\alpha-\beta\| \leq \delta$. Then $\|Q_{\alpha}-Q_{\beta}\|_{\infty} \leq 2 \delta $. As a consequence, 

\begin{enumerate}
\item 
$\alpha\mapsto Q_{\alpha}$ is continuous and $\alpha\mapsto \mathcal{E}_{\alpha}$ is lower semicontinuous in the sense of Theorem \ref{lowersemicont}.
\item 
If $\epsilon-2\delta>0$, then $\mathcal{E}_{\epsilon+2\delta,\alpha}\subseteq\mathcal{E}_{\epsilon,\beta}\subseteq \mathcal{E}_{\epsilon-2\delta,\alpha}$.
\end{enumerate}
\end{prop}

\begin{proof}

For any $a\in (A_{x})_{1}, b\in (A_{y})_{1}$, we have

\begin{align*}
\| \alpha(a)b-b\alpha(a)\|&\le \|\alpha(a)b-\beta(a)b\|+\|b\alpha(a)-b\beta(a)\|+\|\beta(a)b-b\beta(a)\|\\
&\le 2\delta + \|\beta(a)-b\beta(a)\|
\end{align*}

Taking the supremum over $a\in (A_{x})_{1}, b\in (A_{y})_{1} $, we obtain

$$Q_{\alpha}(x,y)\le Q_{\beta}(x,y)+2\delta$$

By switching $\alpha$ and $\beta$ and putting these together, we obtain

$$||Q_{\alpha}-Q_{\beta}||\le 2\delta.$$

If $(x,y)\in E^{\alpha}_{\epsilon}$, then $Q_{\alpha}(x,y)\ge Q_{\alpha}(x,y)-2\delta$, and thus $(x,y)\in E^{\beta}_{\epsilon-2\delta}$. Thus 

$$\mathcal{E}_{\epsilon, \alpha}=\langle E^{\alpha}_{\epsilon}\rangle\subseteq \langle E^{\beta}_{\epsilon-2\delta} \rangle=\mathcal{E}_{\epsilon-2\delta, \beta}.$$

The other inclusion is obtained by switching $\alpha$ and $\beta$.

\end{proof}

Another natural question to ask in the dynamical setting which we did not see for correlations is, given two dynamics $\alpha,\beta:A\rightarrow A$, what happens when we compose them? In particular what is $\mathcal{E}_{\alpha\circ \beta}$, in terms of $\mathcal{E}_{\alpha}$ and $\mathcal{E}_{\beta}$? It's clear that the coarse structure can certainly get smaller. For example, for a spin system $A:=\otimes_{x\in X} M_{n}(\mathbbm{C})$, $\mathcal{E}_{\alpha\circ \alpha^{-1}}=\mathcal{E}_{\text{Id}_{A}}$, which reduces to the trivial diagonal coarse structure. Thus we expect it would be difficult to give a precise characterization of $\mathcal{E}_{\alpha\circ \beta}$ without specific, detailed knowledge of the channels.

Nevertheless, if two channels are local with respect to a coarse structure $\mathcal{E}$, we would expect their composition to also be local with respect to $\mathcal{E}$. This leads us to guess that $\mathcal{E}_{\alpha\circ \beta}\subseteq \langle \mathcal{E}_{\alpha}\cup \mathcal{E}_{\beta}\rangle$. Indeed, we have the following:

\begin{thm}
Let $A=\otimes_{x\in X} M_{n}(\mathbbm{C})$ be a spin system, and let $\alpha,\beta\in \mathcal{B}_{1}(A\rightarrow A)$. Then $\mathcal{E}_{\alpha\circ \beta}\subseteq \langle \mathcal{E}_{\tilde{\alpha}}\cup \mathcal{E}_{\beta}\rangle$. In particular if $\alpha$ has higher order commutator decay, then $\mathcal{E}_{\alpha\circ \beta}\subseteq \langle \mathcal{E}_{\alpha}\cup \mathcal{E}_{\beta}\rangle$.
\end{thm}

\begin{proof}

Let $\epsilon>0$. Choose let $0<\delta<\frac{\epsilon}{3}$. There exists a controlled set $F\in \mathcal{E}_{\beta}$ such that for $x,y\notin F$, $||[\beta(a),b]||<\delta$ for $a\in (A_{x})_{1}$ and $b\in (A_{y})_{1}$. This implies by \cite{MR3050163} that if we set $G=F[x]$, $$||\beta(a)-E_{G}(\beta(a))||<2\delta,$$
\noindent where $E_{G}: A\rightarrow A_{G}$ is the canonical conditional expections, hence

$$\|[\alpha(\beta(a)),b]\|\le\|[\alpha(E_{G}(\beta(a))),b]\|+\|[\alpha(\beta(a)-E_{G}(\beta(a))),b]\|\le \|[\alpha(E_{G}(\beta(a))),b]\|+2\delta $$

\noindent But there exists some controlled set $\widetilde{F}\in \mathcal{E}_{\tilde{\alpha}}$ with $\|[\alpha(E_{G}(\beta(a))),b]\|<\delta$ for $x,y\notin \widetilde{F}$. Therefore, for $x,y\notin F\cup \widetilde{F}$ we have

$$\|[\alpha(\beta(a)),b]\|\le 3\delta<\epsilon$$

\noindent But $F\cup \widetilde{F}\in \langle \mathcal{E}_{\alpha}\cup \mathcal{E}_{\beta}\rangle$, hence $\alpha\circ \beta$ is local with respect to the latter.
\end{proof}


 \begin{rem}\textbf{Choice of tensor product decomposition for spin systems}. We begin this subsection by examining the choice of site structure for dynamics on a spin system. We begin by noting that if $\phi, \psi\in \mathcal{B}_{1}(A\rightarrow A)$ and there eixsts an automorphism $\alpha \in \text{Aut}(A)$ such that $\phi=\alpha^{-1}\circ \psi \circ \alpha$, then if we define the new coarse structure $\{A^{\prime}_{x}:=\alpha(A_{x})$, then $Q_{\phi}=Q^{\prime}_{\psi}$. Indeed, since automorphisms are norm isometries,

\begin{align*}
    Q_{\phi}(x,y)&=\text{sup}_{a\in (A_{x})_{1},b\in (A_{y})_{1}}||\phi(a)b-b\phi(a)||\\
    &=\text{sup}_{a\in (A_{x})_{1},b\in (A_{y})_{1}}||\alpha^{-1}(\psi(\alpha(a)))b-b\alpha^{-1}(\psi(\alpha(a)))||\\
    &=\text{sup}_{a\in (A_{x})_{1},b\in (A_{y})_{1}}||\psi(\alpha(a))\alpha(b)-\alpha(b)\psi(\alpha(a))||\\
    &=\text{sup}_{a^{\prime}\in (A^{\prime}_{x})_{1},b^{\prime}\in (A_{y})_{1}}||\psi(a^{\prime})b^{\prime}-b^{\prime}\psi(a^{\prime})||\\
    &=Q^{\prime}_{\psi}(x,y).
\end{align*} 

This is a very similar situation to the case of correlation structures discussed in \ref{Tensor prod 1}. However, one major difference is that for spin systems with $A=\otimes_{x\in X} M_{d}(\mathbbm{C})$, as far as we know there is not a complete classification of automorphisms up to conjugacy.  However, similarly to the case of states, we expect that essentially any coarse structure can be realized via a given automorphism simply by picking an alternate site structure.
\end{rem}

\subsection{Algebraic entanglement}\label{OtherCoarse}

In this section, we discuss another kind of coarse structure that is specifically interesting for nets of algebras that are \textit{not} concrete spin systems.

 For any concrete spin system, the local observables algebras factorize as a product of algebras assigned to points, i.e. $A_{F}\cong \otimes_{x\in F} A_{x}$. However, for most examples of nets of algebras, this is not the case. Operators that are not products of operators localized in proper subregions can arise as the regions get larger. We describe this situation as \textit{algebraic entanglement}. However, we would expect that for regions $F$ and $G$ that are sufficiently far apart, new operators \textit{should not} emerge which are localized in $F\cup G$ that are not generated by operators localized in $F$ and $G$ respectively.

 Given an abstract net on $X$ whose local algebras are finite dimensional, we can measure algebraic entanglement with the function $X\times X\rightarrow \mathbbm{N}\cup \{0\}$ by

$$D(x,y):=\text{dim}(A_{x\cup y})-\text{dim}(\langle A_{x}, A_{y}\rangle).$$

We can then define the \textit{algebraic entanglement coarse structure} as the universal coarse structure

$$\mathcal{E}_{ae}(A)=\mathcal{E}_{D}.$$ in the sense of Theorem \ref{univ. cs}

\begin{ex}\textbf{Fusion spin chains}: Fusion spin chains are nets of algebras typically defined on $\mathbbm{Z}$ that arise in the context of 2+1D topologically ordered spin systems and categorical dualities of spin chains. Given a fusion category $\mathcal{C}$ and a strongly tensor generating object $X$, there is an associated net of algebras $\mathcal{A}(\mathcal{C},X)$ over $\mathbbm{Z}$. We refer the reader to \cite{Jones2023DHRBO}, \cite{JL24}. It is easy to see that $\mathcal{E}_{ae}(\mathcal{A}(\mathcal{C},X))$ is the standard metric coarse structure on $\mathbbm{Z}$.
\end{ex}

\begin{ex}
    Suppose we have a net $A$ of C*-algebras over $X$. Suppose further that we are given a metric $d$ on $X$. Choose a $t\ge 0$. Then we can define a \textit{refinement} $A^{d,t}$ of $A$ by setting

    $$A^{d,t}_{F}=\langle x\in A_{F}\ :\ x\in A_{G}, G\subseteq F\ , \text{diam}(G)\le t\}$$

    This assmebles into a subnet of $A$, with the property that the algebraic entanglement coarse structure  $\mathcal{E}_{ae}(A^{d,t})$ is contained in the metric coarse structure $\mathcal{E}_{(X,d)}$. Indeed, if $d(x,y)>t$, then $A^{d,t}_{x\cup y}=\langle A^{d,t}_{x}, A^{d,t}_{y}\rangle $ by definition hence $D(x,y)=0.$ Thus the set $E_{k}=\{(x,y)\in X\times X\ : D(x,y)\ge k\}\subseteq \{(x,y)\ :\ d(x,y)\le t\}$, and hence is controlled in the metrix coarse structure of $(X,d)$.
\end{ex}

\section{Stability of coarse structures under quasi-local perturbation} \label{stability of cs}

Local perturbations are operations that only change your system in a small region. Quasi-local perturbations are norm limits of these operations. Stability of a property (of a state, Hamiltonian, or dynamics) on a spin system under quasi-local perturbation is an important criterion for the property to be considered macroscopic. Our goal will be to show that the universal correlation coarse structures defined for states and quantum dynamics are invariant under quasi-local perturbation under favorable circumstances. 

Typically quasi-local perturbations are considered for Hamiltonians, hence we need to give a formal definition for states and dynamics directly.

\begin{defn}\label{quasilocalperturbation} Let $A$ be a discrete net of C*-algebras over $X$. A unital, completely positive map $\Psi:A\rightarrow A$ is called \textit{localized} if there exists a finite subset $F\subseteq X$ such that $$\Psi|_{A_{F^{c}}}=\text{Id}|_{A_{F^{c}}}.$$

A channel $\Psi$ is a \textit{quasi-local perturbation} if there exists a sequence $\Psi_{n}$ of localized channels such that $\|\Psi-\Psi_{n}\|\rightarrow 0$.
\end{defn}

Note that the norm we use for channels is the \textit{uniform norm}, i.e. $\|\Phi\|=\sup_{a\in (A)_{1}}\|\Phi(a)\|$.


\begin{prop}\label{prop:inner auto}
Every inner automorphism of a discrete net $A$ of C*-algebras is a quasi-local perturbation.
\end{prop}

\begin{proof}
Let $u\in U(A)$. Choose an increasing sequence of finite subsets $F_{1}\subseteq F_{2}\subseteq \dots$ such that $\bigcup_{n} F_{n}=X$. 

It is well known (using standard C*-algebra arguments, for example see the proof of Lemma 3.1 in \cite{MR0112057}) that there exists unitaries $u_{i}\in A_{F_{i}}$ such that $\|u_i-u\|\rightarrow 0$. But each $\text{Ad}(u_i)$ is localized in the bounded set $F_i$. If we choose $i$ sufficiently large such that $\|u-u_i\|<\epsilon$, then for $\|a\|\le 1$, we have \begin{align*}
\|uau^{*}-u_iau^{*}_i\|&=\|u_iau^{*}+(u-u_i)au^{*}-u_iau^{*}_i\|\\
&\le \epsilon+ \|u_ia(u-u_{i})\|\le 2\epsilon.
\end{align*}

hence $$\|\text{Ad}(u_i)-\text{Ad}(u)\|\rightarrow 0$$

\end{proof}

The next result shows that for nice states, the universal 2-point correlation coarse structure is stable under quasi-local perturbation.

\begin{thm}\label{full stability} Let $A$ be a discrete net of C*-algebras over the set $X$, and $\phi:A\rightarrow B$ be any ucp map such that $\mathcal{E}_{\phi}$ is proper and monogenic. Then for any quasi-local perturbation $\Psi:A\rightarrow A$, $(\mathcal{E}_{\phi\circ \Psi})_{con}=(\mathcal{E}_{\phi})_{con}$. 
\end{thm}

\begin{proof}

Choose an increasing sequence $F_{1}\subseteq F_{2} \dots $ of finite subsets with $\Psi_{i}$ localized in $F_{i}$ and $\|\Psi_{i}-\Psi\|\rightarrow 0$. Note that $\|\phi\circ \Psi-\phi\circ\Psi_{i}\|\rightarrow 0$.

First we claim that $(\mathcal{E}_{\phi})_{con}=(\mathcal{E}_{\phi\circ \Psi_{i}})_{con}$. To see this, recall that if $\mathcal{E}$ is a proper coarse structure, then for any finite subset $C\subseteq X$, $\mathcal{E}_{con}=\langle \mathcal{E}|_{(X-C)\times (X-C)}\bigcup \mathcal{F}(X\times X)\rangle$, where as usual $\mathcal{F}(X\times X)$ denotes the collection of finite subsets of $X\times X$, and $\mathcal{E}|_{(X-C)\times (X-C)}=\mathcal{E}\cap \left((X-C)\times (X-C)\right)$.

We note that since $\Phi_{i}$ is localized in $F_{i}$,

$$C_{\phi}|_{(X-F_{i})\times (X-F_{i})}=C_{\phi\circ \Psi_{i}}|_{(X-F_{i})\times (X-F_{i})}$$

Now note that for any finite set $C\subseteq X$, a coarse structure $\mathcal{E}$ on $X$ is proper if and only if the coarse structure $\mathcal{E}|_{(X-C)\times (X-C)}$ is proper. Indeed if $\mathcal{E}$ is proper, clearly it's restriction is proper. If $\mathcal{E}$ is not proper, there is an infinite, bounded subset $D\subseteq X$. But then $D\cap (X-C)\subseteq X-C$ is infinite since $C$ is finite, but is bounded in $\mathcal{E}|_{(X-C)\times (X-C)}$. Thus $\mathcal{E}|_{(X-C)\times (X-C)}$ is not proper.

This implies 

$$\mathcal{E}_{\phi}|_{(X-F_{i})\times (X-F_{i})}=\mathcal{E}_{\phi\circ \Psi_i}|_{(X-F_{i})\times (X-F_{i})},$$ hence $\mathcal{E}_{\phi\circ \Psi}$ is proper, and thus we have 

$$(\mathcal{E}_{\phi})_{con}= \langle (\mathcal{E}_{\phi})|_{(X-F_i)\times (X-F_i)} \cup \mathcal{F}(X\times X)\rangle=\langle (\mathcal{E}_{\phi\circ \Psi_{i}})|_{(X-F_i)\times (X-F_i)} \cup \mathcal{F}(X\times X)\rangle=(\mathcal{E}_{\phi\circ \Psi_i})_{con}$$.

Furthermore, if we define for any state $\psi$ the set $E^{F_i}_{\epsilon, \psi}=\{(x,y)\in (X-F_{i})\times (X-F_{i})\ :\ C_{\psi}(x,y)\ge \epsilon\}$, we have 

$$E^{F_i}_{\epsilon, \phi}=E^{F_i}_{\epsilon, \phi\circ \Psi_{i}}$$

\noindent and thus if $\phi$ is stable, for $\epsilon$ in the stable range we have 

$$(\mathcal{E}_{\phi\circ \Psi_i})_{con}=(\mathcal{E}_{\phi})_{con}=\langle E^{F_i}_{\epsilon, \phi}\cup \mathcal{F}(X\times X)\rangle= \langle E^{F_i}_{\epsilon, \phi\circ \Psi_{i}}\cup \mathcal{F}(X\times X)\rangle.$$

\noindent Thus $\phi\circ \Psi_{i}$ is connectively stable, with the same stable range as $\phi$.

Now, let $\epsilon>0$ in the interior of the stable range for $\phi$. Choose $i$ such that $\delta=3||\phi\circ \Psi-\phi\circ \Psi_i||$ satisfies $\epsilon-\delta>0$ and $\epsilon+\delta$ is in the stable range for $\phi$. Then by Proposition \ref{stability1}, $E_{\epsilon+\delta,\phi\circ \Psi_{i}} \subseteq E_{\epsilon,\phi\circ \Psi}\subseteq E_{\epsilon-\delta,\phi\circ \Psi_{i}}$ 
    
$$(\mathcal{E}_{\phi})_{con}=\langle E_{\epsilon-\delta,\phi\circ \Psi_{i}}\bigcup \mathcal{F}(X\times X)\rangle\subseteq \langle E_{\epsilon, \phi\circ \Psi}\bigcup \mathcal{F}(X\times X)\rangle \subseteq \langle E_{\epsilon+\delta,\phi\circ \Psi_{i}}\bigcup \mathcal{F}(X\times X)\rangle \subset (\mathcal{E}_{\phi})_{con}. $$ 

\noindent Since this is true for all $\epsilon>0$, we have $(\mathcal{E}_{\phi})_{con}=(\mathcal{E}_{\phi\circ \Psi})_{con},$

\noindent giving the desired result.

\end{proof}

\begin{rem} The connectivity in the statement of the theorem cannot be removed. Consider any product state $\phi$, so that  on a quantum spin system $A=\otimes_{x\in X} M_{d}(\mathbbm{C})$ and suppose we take an entangling unitary $u$ localized in some finite region $F$. Then $\mathcal{E}_{\phi}$ is diagonal, while $\mathcal{E}_{\phi\circ \text{Ad}(u)}$ contains non-diagonal finite subsets.
\end{rem}

\begin{cor}\label{normclose}
If $\phi$ is a pure state on $A$ such that $\mathcal{E}_{\phi}$ is monogenic and proper, then for any pure state $\psi$ with $d(\phi,\psi)<2$, we have $(\mathcal{E}_{\phi})_{con}=(\mathcal{E}_{\psi})_{con}$.
\end{cor}

\begin{proof}
If two pure states are close, there is a unitary $u\in A$ such that $\psi=\phi\circ \text{Ad}(u)$ \cite{MR115104}. The result then follows from \ref{full stability} and \ref{prop:inner auto}.
\end{proof}

\begin{rem} \textbf{Closure properties for states}. Let $\mathcal{E}$ be monogenic and connected, and let $S(\mathcal{E})$ be the set of pure states whose connected correlation coarse structure is $\mathcal{E}$. Then the above theorem shows that $S(\mathcal{E})$ is a closed subspace of all pure states in the \textit{norm topology} on the state space. However, in general $S(\mathcal{E})$ is \textit{not} closed in the weak $*$-topology. Consider, for example, a discrete net of C*-algebras whose quasi-local algebra $A$ is simple (and norm separable). Then by \cite{MR1994863}, for any two pure states $\phi$ and $\psi$, there exists an approximately inner automorphism $\alpha\in \text{Aut}(A)$ such that $\phi\circ \alpha=\psi$. By approximately inner, we mean there exists a sequence of unitaries $\{u_{i}\in A\}$ such that for all $a\in A$, $||u^{*}_{i}au_{i}-\alpha(a)||\rightarrow 0$. 

Now, if $\phi\in S(\mathcal{E})$ and $\psi\notin S(\mathcal{E})$, let $\alpha$ and $\{u_i\}$ be as above. If $\mathcal{E}$ is monogenic, then $\phi^{u_i}\in S(\mathcal{E})$, but the approximation condition implies $$|\phi(u^{*}_{i}au_{i})-\phi(\alpha(a))|=|\phi(u^{*}_{i}au_{i})-\psi(a)|\rightarrow 0,$$ hence $\phi^{u_i}\rightarrow \psi$ in the weak*-topology. But $\psi$ was arbitrary, so the set $S(\mathcal{E})$ is dense in the pure state space. As a consequence, we do not expect the assignment $\phi\mapsto \mathcal{E}_{\phi}$ (or any variation, e.g. $(\mathcal{E}_{\phi})_{con}$) to be continuous in any sense with respect to the weak*-topology, in contrast to the case of the norm metric of Corollary \ref{normclose}.

\end{rem}


Now we state a similar stability result for the quantum dynamical coarse structure.

\begin{thm}\label{thm: dynamicalcoarsestability} Let $A$ be a discrete net of C*-algebras over the set $X$, and $\phi:A\rightarrow B$ be any ucp map such that $\mathcal{E}_{\phi}$ is proper and monogenic. Then for any quasi-local perturbation  $\Psi:A\rightarrow A$, $(\mathcal{E}_{\phi\circ \Psi})_{con}=(\mathcal{E}_{\phi})_{con}$. 
\end{thm}

\begin{proof}
The proof exactly follows the proof of \ref{full stability} with $Q_{\alpha}$ in place of $C_{\phi}$, and Proposition \ref{dynamicalclos} in place of Proposition \ref{stability1}.
\end{proof}

\begin{cor}\label{cor:dynamicalcorrollary}If $A$ is simple and $\mathcal{E}_{\alpha}$ is monogenic and proper for $\alpha\in \text{Aut}(A)$, then if $\beta\in \text{Aut}(A)$ such that $\|\alpha-\beta\|<1$, then $(\mathcal{E}_{\alpha})_{\text{con}}=(\mathcal{E}_{\beta})_{\text{con}}$. 
\end{cor}
\begin{proof}
 If $||\alpha-\beta||<1$, then $||\beta^{-1}\circ\alpha-\text{Id}_{A}||<1$, and using the power series expansion, we have a derivation $\delta=\text{log}(\beta^{-1}\circ \alpha)$. But since every derivation of a simple C*-algebra is inner \cite{MR223906}, there eixsts some skew-adjoint $h\in A$ with $\delta(a)=[h,a]$. In particular, $\beta^{-1}\circ\alpha(a)=\text{exp}(h)a\text{exp}(-h)$, hence $\alpha=\text{Ad}(u)\circ \beta$, where $u=\beta(\text{exp}(h))$.
\end{proof}

\section{Locality preserving channels}\label{loc pres chan}

In this section, let $X$ be a discrete set and let $\mathcal{E}$ be a fixed coarse structure on $X$. In this section, we will study the collection of quantum channels which in some sense ``preserve coarse locality".

\begin{defn}
A unital completely positive map $\Psi:A\rightarrow B$ has \textit{controlled spread} if for all $\epsilon>0$, there exists a controlled set $E$ such that for any bounded subset $F\subseteq X$ and $a\in (A_{F})_{1}$, there exists $b\in (B_{E[F]})_{1}$ such that $||\Psi(a)-b||<\epsilon$.
\end{defn}

Note that in practice, we will typically be interested in connected, proper coarse spaces, so that bounded sets are precisely finite sets. 

\begin{ex}\textbf{Bounded spread channels}.
Controlled spread generalizes the notion of \textit{bounded spread}, which is simply the above condition for $\epsilon=0$. Thus $\Psi:A\rightarrow A$ with controlled spread are generalizations of (quantum) cellular automata that are possibly stocashtic and possibly ``have tails". 
\end{ex}

\begin{ex}\textbf{Higher order commutator decay}.
If $A:=\otimes_{x\in X} M_{n}(\mathbbm{C})$ is a quantum spin system, if $\Psi: A\rightarrow A$, then if $\widetilde{Q}_{\Psi}$ has higher order controlled decay with respect to $\mathcal{E}$, then $\Psi$ has controlled spread. This follows immediately from \cite{MR3050163}. In particular, time evolutions of local Hamiltonian have controlled spread.
\end{ex}

\begin{defn}\label{locality pres}
A unital completely positive map $\Phi:A\rightarrow B$ is called \textit{locality preserving} if

\begin{enumerate}
\item 
$\Phi$ has controlled spread. 
\item
$\widetilde{C}_{\Phi}$ has higher order controlled decay. 
\end{enumerate}
\end{defn}

Note that if $\Phi$ is a homomorphism (e.g. unitary dynamics), then $\widetilde{C}_{\Phi}$ is identically $0$, hence the second condition is redundant. This condition is for quantum channels which are not multiplicative, but are asymtptoically multiplicative.

\begin{thm}\label{comp of loc pres}
If $A, B, C$ are all discrete nets over $X$, and $\Psi:A\rightarrow B$ and $\Phi:B\rightarrow C$ are locality preserving, then so is $\Phi\circ \Psi$
\end{thm}

\begin{proof}
Let $\epsilon>0$. Choose $E, E^{\prime}\in \mathcal{E}$ implementing the controlled spread condition for $\frac{\epsilon}{2}$ for $\Psi$ and $\Phi$ respectively. Set $E^{\prime \prime}=E^{\prime}\circ E$. Let $F\subseteq X$ be bounded. For $a\in A_{F}$, there exists $b\in B_{E[F]}$ with $\|\Psi(a)-b\|<\frac{\epsilon}{2}$ and $c\in C_{E^{\prime}[E[F]]}=C_{(E^{\prime}\circ E)[F]}$ such that $\|\Phi(b)-c\|<\frac{\epsilon}{2}$, hence 

$$\|\Phi\circ \Psi(a)-c\|\le \|\Phi(b)-c\|+\frac{\epsilon}{2}<\epsilon.$$

\noindent Thus $\widetilde{C}_{\Phi\circ \Psi}$ has controlled spread.

Now, let $E\in \mathcal{E}$. Let $a\in (A_{E[x]})_{1}, a^{\prime}\in (A_{E[y]})_{1}$. Let $\epsilon>0$ and choose $\delta$ such that $\delta^{2}+10\delta<\epsilon$. Then there is some controlled set $E^{\prime}$ implementing the $\delta$ controlled spread condition for $\Psi$, hence there are $b\in B_{(E^{\prime}\circ E)[x]}$ and $b^{\prime}\in B_{(E^{\prime}\circ E)[x]}$ with $\|\Psi(a)-b\|<\delta$ and $\|\Psi(a^{\prime})-b^{\prime}\|<\delta$. Then also, using the $\delta$ controlled spread condition for $\Phi$ as well, we have for some controlled set $E''$ $c\in C_{(E''\circ E')[x]}$ and $c^{\prime}\in C_{(E''\circ E')[x]}$ with $\|\Phi(b)-c\|<\delta$ and $\|\Phi(b^{\prime})-c^{\prime}\|<\delta$. We also pick the higher order controlled decay bounds to be $\delta$ for both $\Psi, \Phi$. We then get,

\begin{align*} &\|\Phi\circ \Psi(aa')-\Phi\circ \Psi(a)\Phi\circ \Psi(a')\|\\
&\le \|\Phi\circ \Psi(aa')-\Phi(b)\Phi(b^{\prime})\| +2\delta+\delta^{2}\\
&= \|\Phi\circ \Psi(aa') -\big\{\Phi(b)[\Phi(b')-c']+\Phi(b)c'+cc'-cc'\big\}\|+2\delta +\delta^2\\
&\leq \delta^2 +4\delta +\|\Phi\big(\Psi(aa')-bb'\big)+\Phi(bb')-cc'\|\\
&\leq \delta^2 +4\delta+ \|\Phi\big(\Psi(aa')-bb'\big)\|+\|\Phi(bb')-cc'\|\\
&\leq \delta^2 +4\delta+ \|\Psi(aa')-bb'\|+\|\Phi(bb')-cc'\|\\
&\leq \delta^2 +4\delta+ \|\Psi(aa')-\Psi(a)\Psi(a^{\prime})+\Psi(a)\Psi(a^{\prime})-bb'\|+\|\Phi(bb')-cc'\|\\
&\leq \delta^2 +4\delta+ \|\Psi(aa')-\Psi(a)\Psi(a^{\prime})\|+\|\Psi(a)(\Psi(a^{\prime})-b^{\prime})\|+\\
& \ \ \ \ \ \|(\Psi(a)-b)b'\|+\|\Phi(bb')-cc'\|\\
&\leq \delta^2+7\delta + \|\Phi(bb')-\Phi(b)\Phi(b')\|+\|\Phi(b)(\Phi(b')-c')\|+\|(\Phi(b)-c)c'\|\\
&\leq \delta^2 +10\delta<\epsilon
\end{align*}

 Hence, $\widetilde{C}_{\Phi\circ \Psi}$ has higher order controlled decay.
\end{proof}

\noindent We now introduce a category $\mathcal{LP}_{\mathcal{E}}$ defined as follows: 

\begin{enumerate}
    \item 
    Objects are discrete nets of C*-algebras over $X$.
    \item 
    Morphisms from $A$ to $B$, denoted $\mathcal{LP}_{\mathcal{E}}(A,B)$ ar locality preserving channels $\Psi:A\rightarrow B$
    \item
    Composition is composition of channels.
\end{enumerate}

We can restrict to the subcategory of invertible morphisms, which consist of \textit{automorphisms} with controlled spread, which is a groupoid we denote by $\mathcal{LP}^{0}_{\mathcal{E}}$. This contains the groupoid $\text{Net}_{X}$ studied in \cite{Jones2023DHRBO}.

Now, the assignment $A\mapsto S(A)$ of the state space of $A$ to it's quasi-local C*-algebra extends to a functor $F:\mathcal{LP}_{\mathcal{E}}^{op}\rightarrow \text{Conv}$, where $\text{Conv}$ denotes the category of convex spaces, defined by 

$$F(\Psi)(\phi)=\phi\circ \Psi,$$

\noindent where $\Psi: A\rightarrow B$ is locality preserving, and $\phi\in S(B)$.

Though we do not pursue this in this paper, a very interesting question is: what are the orbits under the action of the groupoid $\mathcal{LP}^{0}_{\mathcal{E}}$? Can we characterize them with some kind of invariant? This can be interpreted as a kind of state-based classification of ``phases".  

The following theorem shows that locality preserving channels in fact preserve locality of states, as characetrized by higher-order decay of correlations, justifying our terminology.

\begin{cor}
    Let $\mathcal{S}_{\mathcal{E}}(A):=\{\phi\in S(A)\ : \mathcal{E}_{\widetilde{\phi}}\subseteq \mathcal{E}\}$. For any discrete nets $A,B$ over $X$ and $\Psi\in \mathcal{LP}_{\mathcal{E}}(A,B)$, we have $$F(\Psi)S_{\mathcal{E}}(B)\subseteq S_{\mathcal{E}}(A)$$
\end{cor}

\begin{proof}
Choose $\phi\in S_{\mathcal{E}}(B)$ and $\Psi\in \mathcal{LP}_{\mathcal{E}}(A,B)$ Let $\epsilon>0$. Choose $\delta>0$ such that $10\delta+\delta^{2}<\epsilon$, and let $E\in \mathcal{E}$. Then there exists a symmetric $H\in \mathcal{E}$ such that for $a\in (A_{E[x]})_{1}$ and $a^{\prime}\in (A_{E[y]})_{1}$, there are $b\in B_{H[E[x]]}, b^{\prime}\in B_{H[E[y]]}$ such that $$\|\Psi(a)-b\|<\delta$$

and the same bound for $a'$ and $b'$. Since $(\widetilde{C}_{\phi})_{H\circ E}$ has controlled decay with respect to $\mathcal{E}$, there exists some $F$ (which we can assume contains $H$ by taking union) such that for $x,y\notin F$, $(\widetilde{C}_{\phi})_{H}(x,y)< \epsilon-(10\delta+\delta^{2})$. Thus for $(x,y)\notin F$, and by thinking of $\phi$ as a locality presreving channel to the scalars (i.e thinking of the scalars as a trivial discrete net of algebras), it follows from Theorem \ref{comp of loc pres} that

$$(\widetilde{C}_{\phi\circ \Psi})_{F}(x,y)< \epsilon.$$

Thus $\mathcal{E}_{\widetilde{\phi\circ \Psi}}\subseteq \mathcal{E}$.
\end{proof}

We record the following corollary, and its implications for the complexity of state preparation on a quantum processor.

\begin{cor} \label{dynamics and states}
    Let $\phi$ be a product state on a spin system $A=\otimes_{x\in X} M_{d}(\mathbbm{C})$. Let $\alpha\in \text{Aut}(A)$. Then $\mathcal{E}_{\widetilde{\phi\circ \alpha}}\le \mathcal{E}_{\widetilde{\alpha}}$. In particular, if $\alpha$ is local with respect to $\mathbbm{Z}_{n}$ then $\mathcal{E}_{\widetilde{\phi\circ \alpha}}$ is not equivalent to $\mathbbm{Z}^{m}$ for any $m>n$.
\end{cor}

If we interpret the universal coarse structures of states and dynamics as measures of ``geometric complexity", the above corollary shows that to prepare a state with a certain geometric complexity, we need a circuit which is at least as geometrically complex. In particular, if we want to prepare a state with $d$-dimensional correlations, we need a circuit which is local to some coarse structure containing $\mathbbm{Z}^{d}$.

\bigskip

\section{Coarse dependence of order parameters}\label{coarsedependence}

In this section, we investigate to what extend entanglement-based order parameters in fact depend only on the coarse equivalence class of a metric rather than the specific details of the metric itself. An order parameter that only depends on the coarse equivalence class of a metric will be called a \textit{coarse order parameter}. Here we focus on two examples: topological order and critical exponents. In both of these examples, it is not in fact obvious that these features depend only on the coarse structure. The goal of this section is to provide a rigorous proof of this fact.

\subsection{Topological order: superselection sectors}\label{topord} First we will investigate topological order, utilizing the mathematical formulation for spin systems following \cite{MR2804555, MR2956822, MR3426207, MR3764565, MR4362722}. We will focus on topological order as a W*-category of superselection sectors. In general this is only part of the structure of topological order: in dimensions larger than 1, the category of cone localized superselection sectors (under some assumptions on the state) can be equipped with the additional structure of a braided C*-tensor category \cite{MR4362722}. In higher dimensions this braided category is only a piece of the full structure of topological order, which should assemble into a \textit{higher} braided category \cite{MR4444089}. Nevertheless, we expect the analysis demonstrated here to apply in more complicated settings, and thus topological order should depend only on the quasi-isometry class of the metric, rather than its isometry or Lipschitz class.

Suppose we have a concrete spin system with quasi-local algebra $A=\otimes_{x\in \mathbbm{Z}^{n}} M_{q}(\mathbbm{C})$, and let $\phi$ be a state on $A$. Let $H_{\phi}$ be the GNS representation of $A$ associated to $\phi$. We make of use of the notion of \textit{quasi sub-equivalence} of representations of two C*algebras. If $H,K$ are two representations of a unital C*-algebra $B$ we say $H\preceq K$ if $H$ is unitarily equivalent to a summand of $K\otimes \ell^{2}(\mathbbm{N})$.

We recall the usual formulation of superselection sectors in terms of cones in $\mathbbm{R}^{n}$ (or rather, a slight modification as in \cite{jones2023localtopologicalorderboundary}). There are many possible variations of superselection sectors, but for the sake of staying close to \cite{MR2804555, MR4362722}, a \textit{cone} $\mathcal{C}\subseteq \mathbbm{R}^{n}$ will be convex closure of a set $\{r_{i}\}_{i\in I}$ of geodesic rays emanating from a common point, whose interior has dimenion $n$.
By a \textit{cone} in the lattice $\mathbbm{Z}^{n}$, we actually mean the intersection $\mathcal{D}=\mathcal{C}\cap \mathbbm{Z}^{n}$, where $\mathcal{C}$ is a cone in $\mathbbm{R}^{n}$.

\begin{defn}
Let $H$ be a Hilbert space representation of $A$. Then for any subset $F\subseteq X$, we say $H$ is \textit{localizable in F} if $$H|_{A_{F^{c}}}\preceq (H_{\phi})|_{A_{F^{c}}}.$$
\end{defn}

\begin{defn}
  A Hilbert space representation $H$ of $A$ satisfies the \textit{superselection criterion} if for any cone $\mathcal{D}$ in $\mathbbm{Z}^{n}$, $H$ is localizable in $\mathcal{D}$. We denote by $\text{Sec}_{\phi}$ the W*-category of Hilbert space representations of $A$ satisfying the superselection criteria.
\end{defn}

We note that for many states (in fact, all ground states of gapped Hamiltonians) we can replace $\prec$ in the above definition with unitary equivalence without loss of generality, and this recovers the usual definition. $\text{Sec}_{\phi}$ is a unitarily Cauchy complete W*-category \cite{MR808930}.

Our goal in this section is to define a version of $\text{Sec}_{\phi}$ which makes sense for a spin system over any metric space $(X,d)$ and depends only on its quasi-isometry class.

So let $(X,d)$ be a metric space with finite asymptotic dimension $n$, and suppose $(X,d)$ is quasi-geodesic. Let $\mathcal{D}$ be a cone in $\mathbbm{R}^{n}$. For $L>0, C\ge 0$, a coarse $L,C,\mathcal{D}$ cone is a quasi-isometric embedding $f:\mathcal{D}\rightarrow (X,d)$, with Large-scale Lipschitz constants $L,C$. The asymptotic dimension assumption is simply to guarantee we are using the ``right" dimension of cone for the region, though it is not strictly necessary below. Recall also that for a subset $E\subseteq X$, $E+R$ is simply the $R$-neighborhood of $E$ \textit{and does not refer to translation by $R$} which does make sense in general.

Now let $A=\otimes_{x\in X} M_{q}(\mathbbm{C})$ be the quasi-local algebra of a spin system as above, and $\phi$ a state on $A$.

\begin{defn} A Hilbert space representation of $A$ satisfies the \textit{coarse superselection criterion} on $(X,d)$ if for any triple $L,C, \mathcal{D}$, there exists some constant $R\ge 0$ such that for an coarse $L,C,\mathcal{D}$ cone $f:\mathcal{D}\rightarrow X$, $H$ is localizable in $f(\mathcal{D})+R$. We denote by $\text{Sec}^{d}_{\phi}$ the W*-category of Hilbert space representations of $A$ satisfying the coarse superselection criteria.
\end{defn}

The following proposition shows that this coarse superselection category as defined above agrees with the usual superselection category, and thus is intrinsically a property of a state on a quantum processor.

\begin{prop}\label{topordprop} Let $(X,d)$ be a quasi-geodic metric space with finite asymptotic dimension, and let $A=\otimes_{x\in X} M_{q}(\mathbbm{C})$ be a spin system over $X$ and $\phi$ a state. Then
\begin{enumerate}
    \item 
    If $d^{\prime}$ is a metric on $X$ quasi-isometric to $d$, then $\text{Sec}^{d}_{\phi}=\text{Sec}^{d^{\prime}}_{\phi}$.
    \item 
    If $(X,d)=\mathbbm{Z}^{n}$ with the induced Euclidean metric,  then $\text{Sec}^{d}_{\phi}=\text{Sec}_{\phi}$.
\end{enumerate}
\end{prop}

\begin{proof}
To prove the first point, suppose there are constants $L>0,C\ge 0$ with 

$$L^{-1}d(x,y)-C\le d^{\prime}(x,y)\le Ld(x,y)+C,$$

\noindent for all $x,y\in X$. Then the identity $\text{Id}_{X}:X\rightarrow X$ thought of as a map from $(X,d)$ to $(X,d^{\prime})$, is an $L,C$ quasi-isometry.

Then any coarse $K,B, \mathcal{D}$ cone with respect to $d$ is a coarse $LK, LB+C, \mathcal{D}$ cone with respect to $d^{\prime}$ and conversely, $K,B, \mathcal{D}$ cones with respect to $(X,d^{\prime})$ are $L^{-1}K, L^{-1}B-C, \mathcal{D}$ cones with respect to $d$. Now let $H\in \text{Sec}^{d}_{\phi}$, and let $K,B,\mathcal{D}$ and $f: \mathcal{D}\rightarrow (X,d^{\prime})$ be a coarse cone. Then there is some $R\ge 0$ such that $H$ is localizble in $f(\mathcal{D})+R$. Then
for any choice of $R^{\prime}>LR+C$, 

$d^{\prime}(x,f(\mathcal{D}))>R^{\prime}$ which implies $d(x,f(\mathcal{D}))\ge L^{-1}R^{\prime}-C>R$. Thus $$(f(\mathcal{D})+_{d^{\prime}} R^{\prime})^{c}\subseteq (f(\mathcal{D})+_{d}R)^{c},$$ 

\noindent Therefore $A_{(\mathcal{D}+_{d^{\prime}}R^{\prime})^{c}}$ is contained in $A_{(\mathcal{D}+_{d}R)^{c}}$ and the superselction criterion restricts to

$$H|_{A_{(\mathcal{D}+_{d^{\prime}}R^{\prime})^{c}}}\preceq (H_{\phi})|_{A_{(\mathcal{D}+_{d^{\prime}}R^{\prime})^{c}}}.$$

\noindent Thus $\text{Sec}^{ d}_{\phi}\subseteq \text{Sec}^{ d^{\prime}}_{\phi}$. Since the quasi-isometry relation is symmetric, we obtain the other inclusion switching the roles of $d$ and $d^{\prime}$.

To prove the second point, suppose $(X,d)=\mathbbm{Z}^{n}$ with the restriction of the Euclidean metric. Note that for any $K,B, \mathcal{D}$, there is some $R\ge 0$ such that the $R$-neighborhood of the image of any $K,B,\mathcal{D}$ coarse cone contains an honest Euclidean cone (i.e. the intersection of a cone in $\mathbbm{R}^{n}$ with the lattice $\mathbbm{Z}^{n}$). Now let $f:\mathcal{D}\rightarrow \mathbbm{Z}^{n}$ be a coarse $K,B,\mathcal{D}$ cone. If $H\in \text{Sec}_{\phi}$, then since $H$ is localizable in $\mathcal{D}^{\prime}$, since $\mathcal{D}^{\prime}\subseteq f(\mathcal{D})+R$, $H$ is also localizble in $f(\mathcal{D})+R$.  Since $K,B, \mathcal{D}$ and $f$ were arbitrary, $H\in \text{Sec}^{d}_{\phi}.$

Now let $H\in \text{Sec}^{ d}_{\phi}$, and $\mathcal{D}$ an honest cone. If $\mathcal{D}=\mathbbm{Z}^{n}\cap \mathcal{D}^{\prime}$ for $\mathcal{D}^{\prime}$ a cone in $\mathbbm{R}^{n}$, then $\mathcal{D}$ is $1,\sqrt{2}, \mathcal{D}^{\prime}$ cones, where $f:\mathcal{D}^{\prime}\rightarrow \mathcal{D}$ is the map that chooses the nearest element in $\mathcal{D}$. Note that for any vector $v\in \mathbbm{Z}^{n}$, the honest cone $\mathcal{D}^{\prime \prime}=\mathcal{D}+v$ is also a $1, \sqrt{2}, \mathcal{D}^{\prime}$ coarse cone. Let $R$ be the $1, \sqrt{2}, \mathcal{D}^{\prime}$ localizing constant. Then by translating $\mathcal{D}$ into its interior, we can find an honest cone $\mathcal{D}^{\prime \prime}$ such that $\mathcal{D}^{\prime\prime}+R\subseteq \mathcal{D}$. Then $H$ is localized in $\mathcal{D}^{\prime \prime}\subseteq \mathcal{D}$ hence also in $\mathcal{D}$, and thus $H\in \text{Sec}_{\phi}$.

\end{proof}

In the case when $(X,d)$ is coarsely equivalent to $\mathbbm{Z}^{2}$, this category can be upgraded with the structure of a \textit{braided W*-tensor category} under a condition called \textit{approximate Haag duality} \cite{MR4362722}, generalizing the familiar Haag duality from AQFT utilized in earlier approaches \cite{MR2804555, MR2956822, MR3426207, MR3764565}. It is a bit tedious but not difficult to check using the style of argument from our proposition that this condition can be reformulated in terms of coarse cones, and that the construction of the tensor product and braiding can be carried out in the coarse setting. Alternatively, one could show that an apropriate poset of coarse cones satisfies the axioms in \cite{bhardwaj2024superselectionsectorsposetsvon}, which by their main reults would allow a direct definition of braided tensor category structure on $\text{Sec}^{d}_{\phi}$.

\subsection{Coarse decay rates}\label{decay rates}

Although decay rates of functions will depend on the choice of metric, there are certain ``coarse" properties of the decay rate that only depend on the coarse equivalence class. If $\phi$ is a generic equilibrium state for a Hamiltonian with short-range interactions, we expect $C_{\phi}$ to decay at most exponentially as a function of distance, where the constant appearing in the exponent is the inverse correlation length. However, for a state at a critical point in parameter space, we expect the function to decay algebraically, in which case the correlation length is said to diverge. In this setting, there is a notion of \textit{critical exponent}, which is usually used as a fundamental indicator of the states universality class \cite{stanley1987introduction} This discussion motivates the following definitons.

\begin{defn}\label{defn:correlationlength}
Let $f:X\times X\rightarrow \mathbbm{R}_{+}$. Then we say $f$ has \textit{ length 0} if there exists a controlled set $E\in \mathcal{E}_{f}$ such that $f(x,y)=0$ for all $(x,y)\notin E$.

\medskip

If $\mathcal{E}_{f}$ is monogenic, let $d_{\epsilon}$ be the path metric on the graph $G_{\epsilon}$ for some $\epsilon$ in the stable range. Then we say
\begin{enumerate}
\item 
$f$ has \textit{finite length } if there exists positive constants $A,B$ and a controlled set $E\in \mathcal{E}_{f}$ such that $f(x,y)\le A e^{-B d_{\epsilon}(x,y)}$ for all $(x,y)\notin E$.
\item
$f$ \textit{has infinite length } if for every exponential function $g(t)=A e^{-B t}$, there is some controlled set $E$ such that $f(x,y)> g(d_{\epsilon}(x,y))$ for all $(x,y)\notin E$.
\end{enumerate}

\end{defn}

When applied to correlation functions, we can replace the word ``length" by ``correlation length" to recover the familiar notions. Note that $0$ length implies finite length, and finite length implies not infinite  length. Note length $0$ manifestly depends only on the coarse structure, however, at the moment, it is not clear that the notions of finite and infinite length are well-defined, since they depend a-priori on the choice of $\epsilon$ in the stable range. This leads us to the following proposition:

\begin{prop}\label{coarseinvariancecorrelations} Finite and infinite length do not depend on the chocie of $\epsilon$ in the stable range.
\end{prop}

\begin{proof}
Since $f$ is stable, for any choices of $\epsilon, \epsilon^{\prime}$ in the stable range, there exists positive constants $L,C$ such that $L^{-1} d_{\epsilon^{\prime}}(x,y)-C\le d_{\epsilon}(x,y)\le L d_{\epsilon^{\prime}}(x,y)+C$.

Now suppose there are constants $A,B$ such that $f(x,y)\le A e^{-B d_{\epsilon}(x,y)}$ for $(x,y)\notin E$ for some controlled set $E\in \mathcal{E}_{f}$. Then 

$$f(x,y)\le A e^{-B (L^{-1}d_{\epsilon^{\prime}}(x,y)-C)}=(Ae^{C})e^{-(BL^{-1})d_{\epsilon^{\prime}}(x,y)},$$

\noindent as desired. Now suppose $f$ has infinite length. Let $A,B>0$. Then choose $R>0$ so that $d_{\epsilon}(x,y)\ge R$ implies $f(x,y)>Ee^{-Dd_{\epsilon}(x,y)}$, where $E=Ae^{BCL^{-1}}$ and $D=BL^{-1}$. Then since $d_{\epsilon^{\prime}}(x,y)\ge L^{-1}d_{\epsilon}(x,y)-L^{-1}C$, and of exponential function is decreasing
$$Ae^{-Bd_{\epsilon^{\prime}}(x,y)}\le Ee^{-Dd_{\epsilon}(x,y)}<f(x,y)$$
\end{proof}

As a consequence, for any state on a spin system whose coarse structure is monogenic, the finiteness of the correlation length (or in other words, the length of the function $C_{\phi}$) is intrinsic to the state.

We note that in the finite length case, the actual constants $A,B$ depend on the $\epsilon$ you choose, and so are not themselves invariant under the choice of coarse structure. Thus, while we can make asymptotic boundedness statements generally as above (we have focused on the exponential case, one could use other classes of decaying functions), the precise character of asymptotic exponential decay appears to not be coarsely invariant. In what sense, then, can we actually make sense of ``asymptotic rate of decay" that only depends on the quasi-isometry class of the metric $d$? To this end, we introduce the following definitions

\begin{defn}
Let $g,h:\mathbbm{R}_{+}\rightarrow \mathbbm{R}_{+}$. We say $g$ and $h$ are \textit{coarsely commensurate}, written $h\approxeq g$, if for any $L>0$, $C\in \mathbbm{R}$, there exists constants $D_{1},D_{2}>0$ such that for all sufficiently large $x$,

$$D_{1}h(x)<g(Lx+C)<D_{2}h(x).$$

\noindent We say $h$ is \textit{coarsely invariant} if $h$ is coarsely commensurate with itself.
\end{defn}

In the above definition, in order for the inequality to make sense we need to assume at least that $x\ge -\frac{C}{L}$. Note that coarsely simple functions are closed under positive scaling, addition, and multiplication. Clearly $\approxeq$ is an equivalence relation on the collection of coarsely invariant functions.

\begin{defn} . Let $f:X\times X\rightarrow \mathbbm{R}_{+}$. Let $g:\mathbbm{R}_{+}\rightarrow \mathbbm{R}_{+}$ be a coarsely invariant, decreasing function. Then we say $f$ \textit{decays coarsely uniformly} as $g$ if there exist constants $C_{1},C_{2}, R> 0$ such that if $d(x,y)\ge R$,

$$C_{1}g(d(x,y))\le f(x,y)\le C_{2}g(d(x,y)).$$

\end{defn}

\begin{prop}
Let $(X,d)$ be an unbounded quasi-geodesic metric space, and $f:X\times X\rightarrow \mathbbm{R}_{+}$.

\begin{enumerate}
\item
If $g,h: \mathbbm{R}_{+}\rightarrow \mathbbm{R}_{+}$ are coarsely invariant decreasing functions and $f:X\times X \to \mathbbm{R}_{+}$ decays coarsely uniformly as $g$ and $h$, then $g\approxeq h$. Thus coarse uniform decay gives a well-defined coarse commensurability class.
\item
The coarse commensurability class of $f$ only depends on the metric up to coarse equivalence.
\end{enumerate}
\end{prop}

\begin{proof}

Suppose we have coarsely invariant functions $g,h:\mathbbm{R}_{+}\rightarrow \mathbbm{R}_{+}$ and constants $C_{1},C_{2},R>0$ and $D_{1},D_{2}, S>0$ such that for $d(x,y)\ge R$,

$$C_{1}g(d(x,y))\le f(x,y)\le C_{2}g(d(x,y)),$$

and for $d(x,y)\ge S$,

$$D_{1}h(d(x,y))\le f(x,y)\le D_{2}h(d(x,y)).$$

Then for $d(x,y)\ge R+S$, we have 

$$C_{1}g(d(x,y)) \le f(x,y) \le D_{2} h(d(x,y))$$

so for any $t=d(x,y)\ge R+S$ for some $(x,y)\in X\times X$

$$g(t) \le \frac{D_{2}}{C_1} h(t)$$

and similarly for such a $t$,

$$ h(t)\frac{D_{1}}{C_{2}}\le g(t)$$

Since $(X,d)$ is an unbounded, quasi-geodesic metric space there is some $T\ge 0$ such that for all $t\ge R+S$, there exists $t_{l}\le t\le t_{r}$, with $t_{l},t_{r}\in \{d(x,y)\ :\ (x,y)\in X\times X\}$ with $t-t_{l}, t_r-t\le T$.

Using coarse invariance of $g,h$, choose $E_{1},E_{2}>0$ such that $h(t+T)\ge E_{1}h(t)$ and $h(t-T)\le E_{2}h(t)$ for sufficiently large $t$. Then using the fact that $g$ and $h$ are decreasing, for sufficiently large $t$, we have

$$g(t)\ge g(t_r)\ge \frac{D_{1}}{C_2} h(t_{r})=\frac{D_{1}}{C_2} h(t+(t_r-t))\ge \frac{D_{1}}{C_2} h(t+T)\ge \frac{E_{1}D_{1}}{C_2} h(t)$$

$$g(t)\le g(t_{l})\le \frac{D_{2}}{C_1} h(t_{l})=\frac{D_{2}}{C_1} h(t-(t-t_{l}))\le \frac{D_{2}}{C_1} h(t-T)\le  \frac{E_{2}D_{2}}{C_1} h(t)$$

Thus $g$ and $h$ are coarsely commensurate. The second statement follows from the coarse invariance assumption.

\end{proof}

The above proposition gives a well defined meaning for ``decay rate" for functions $f:X\times X\rightarrow \mathbbm{R}_{+}$ that decay coarsely uniformly that only depends on the coarse structure of $(X,d)$, at least in the case that $(X,d)$ is quasi-geodesic. A key feature is that for this to be well-defined, we need the decaying test functions to be coarsely invariant. 

As a counterexample, consider the generic exponential function $g(t)=A\text{exp}(-Bt)$ we consider when defining finite/infinite length. Then  

$$g(Lt+C)=A\text{exp}(-BLt-BC)=A\text{exp}(-BC)\text{exp}(-BLt),$$ which decays much more rapidly than $g(t)$ at infinity if $L\ne 1$ hence there is no constant $D$ such that $Dg(t)\le g(Lt+C)$, so exponential functions are not coarsely stable. What is going on here is that while the \textit{class} of exponential functions is stable under quasi-isometry, hence being bounded (above or below) by \textit{some} exponential is well-defined coarsely. However, no exponential function is itself coarsely invariant in the sense of our definition. Thus while finite length is well defined, the actual value of the length is not.

This leads to the question: which classes of functions appearing in the context of spin systems are coarsely simple, and can we characterize coarse commensurability of these classes?

\begin{prop}\label{criticalexp-welldefined} The functions $\frac{1}{\text{ln}(t)}$ and $\frac{1}{t^{\gamma}}$ for $\gamma>0$ are coarsely invariant. Furthermore, $\frac{1}{t^{\gamma}}\approxeq \frac{1}{t^{\delta}}$ if and only if $\gamma=\delta$.
\end{prop}

\begin{proof}

For the statement about $\text{ln}$, we can see by exponentiating that for sufficiently large $t$, we have for any $A>0, B\in \mathbbm{R}$

$$\frac{1}{2}\text{ln}(t)\le \text{ln}(At+B)\le 2\text{ln}(t),$$

\noindent and inverting and reversing the inequalities gives us the desired outcome.

Now let $\gamma>0$. If $C$ is positive then for any $K>L$, if $t\ge \frac{C}{K-L}$ then $Lt+C\le Kt$, hence

$$L^{\gamma} t^{\gamma} \le (Lt+C)^{\gamma}\le K^{\gamma} t^{\gamma},$$

\noindent and inverting gives the desired inequality. Similarly, if $C<0$, then choosing $K<L$, for $t\ge \frac{-C}{L-K}$ we obtain

$$K^{\gamma} t^{\gamma} \le (Lt+C)^{\gamma}\le L^{\gamma} t^{\gamma},$$

\noindent and inverting gives the desired inequality.

\end{proof}

We find it very interesting that asymptotic algebraic and logarithmic decay of functions, typically witnessed in critical behavior and the correlations of observables in conformal field theories, is a coarse invariant, while asymptotic exponential decay itself is not. One consequence of the above lemma is that for functions with algebraic coarse uniform decay, there is a well-defined analogue of the correlation critical exponent in statistical mechanics, but which only depends on the coarse structure.

\begin{defn} If $X$ is an unbounded, quasi-geodesic metric space, and if $f:X\times X\rightarrow \mathbbm{R}_{+}$ decays coarsely uniformly as $t^{-\gamma}$, then $\gamma$ is called the coarse critical exponent. 
\end{defn}

Our assertion is that the coarse critical exponent (or more generally, the coarse commensurability class of a coarsely stable 2-point correlation function) is universal property of the state.

\bibliographystyle{alpha}
\footnotesize{
\bibliography{bibliography}}

\newcommand{\etalchar}[1]{$^{#1}$}
\begin{thebibliography}{TWH{\etalchar{+}}23}

\bibitem[A{\etalchar{+}}23]{NonAbTop3}
T.~I. Andersen et~al.
\newblock Non-abelian braiding of graph vertices in a superconducting
  processor.
\newblock {\em Nature}, 618(7964):264--269, 2023.

\bibitem[Arr19]{QCAReview1}
P.~Arrighi.
\newblock An overview of quantum cellular automata.
\newblock {\em Natural Computing}, 18(4):885--899, 2019.

\bibitem[Bat02]{MR2038580}
Robert~W. Batterman.
\newblock {\em The devil in the details}.
\newblock Oxford Studies in Philosophy of Science. Oxford University Press, New
  York, 2002.
\newblock Asymptotic reasoning in explanation, reduction, and emergence.

\bibitem[BBC{\etalchar{+}}24]{bhardwaj2024superselectionsectorsposetsvon}
Anupama Bhardwaj, Tristen Brisky, Chian~Yeong Chuah, Kyle Kawagoe, Joseph
  Keslin, David Penneys, and Daniel Wallick.
\newblock Superselection sectors for posets of von neumann algebras, 2024.

\bibitem[BCPH22]{Bluhm2022exponentialdecayof}
Andreas Bluhm, {\'{A}}ngela Capel, and Antonio P{\'{e}}rez-Hern{\'{a}}ndez.
\newblock Exponential decay of mutual information for {G}ibbs states of local
  {H}amiltonians.
\newblock {\em {Quantum}}, 6:650, February 2022.

\bibitem[BHM10]{10.1063/1.3490195}
Sergey Bravyi, Matthew~B. Hastings, and Spyridon Michalakis.
\newblock {Topological quantum order: Stability under local perturbations}.
\newblock {\em Journal of Mathematical Physics}, 51(9):093512, 09 2010.

\bibitem[BMNS12]{AutomorphicEquiv}
Sven Bachmann, Spyridon Michalakis, Bruno Nachtergaele, and Robert Sims.
\newblock Automorphic equivalence within gapped phases of quantum lattice
  systems.
\newblock {\em Communications in Mathematical Physics}, 309(3):835--871, 2012.

\bibitem[BR87]{MR887100}
Ola Bratteli and Derek~W. Robinson.
\newblock {\em Operator algebras and quantum statistical mechanics. 1}.
\newblock Texts and Monographs in Physics. Springer-Verlag, New York, second
  edition, 1987.
\newblock $C^\ast$- and $W^\ast$-algebras, symmetry groups, decomposition of
  states.

\bibitem[BR97]{MR1441540}
Ola Bratteli and Derek~W. Robinson.
\newblock {\em Operator algebras and quantum statistical mechanics. 2}.
\newblock Texts and Monographs in Physics. Springer-Verlag, Berlin, second
  edition, 1997.
\newblock Equilibrium states. Models in quantum statistical mechanics.

\bibitem[CCM17]{MR3773571}
ChunJun Cao, Sean~M. Carroll, and Spyridon Michalakis.
\newblock Space from {H}ilbert space: recovering geometry from bulk
  entanglement.
\newblock {\em Phys. Rev. D}, 95(2):024031, 20, 2017.

\bibitem[CGW10]{PhysRevB.82.155138}
Xie Chen, Zheng-Cheng Gu, and Xiao-Gang Wen.
\newblock Local unitary transformation, long-range quantum entanglement, wave
  function renormalization, and topological order.
\newblock {\em Phys. Rev. B}, 82:155138, Oct 2010.

\bibitem[CM21]{10.3389/frai.2021.667963}
Frédéric Chazal and Bertrand Michel.
\newblock An introduction to topological data analysis: Fundamental and
  practical aspects for data scientists.
\newblock {\em Frontiers in Artificial Intelligence}, 4, 2021.

\bibitem[CNN18]{MR3764565}
Matthew Cha, Pieter Naaijkens, and Bruno Nachtergaele.
\newblock The complete set of infinite volume ground states for {K}itaev's
  abelian quantum double models.
\newblock {\em Comm. Math. Phys.}, 357(1):125--157, 2018.

\bibitem[EM19]{EM}
Eske~Ellen Ewert and Ralf Meyer.
\newblock Coarse geometry and topological phases.
\newblock {\em Communications in Mathematical Physics}, 366(3):1069--1098,
  2019.

\bibitem[Far20]{Farrelly2020reviewofquantum}
Terry Farrelly.
\newblock A review of {Q}uantum {C}ellular {A}utomata.
\newblock {\em {Quantum}}, 4:368, November 2020.

\bibitem[FN15]{MR3426207}
Leander Fiedler and Pieter Naaijkens.
\newblock Haag duality for {K}itaev's quantum double model for abelian groups.
\newblock {\em Rev. Math. Phys.}, 27(9):1550021, 43, 2015.

\bibitem[FU15]{10.1063/1.4921305}
Jürg Fröhlich and Daniel Ueltschi.
\newblock {Some properties of correlations of quantum lattice systems in
  thermal equilibrium}.
\newblock {\em Journal of Mathematical Physics}, 56(5):053302, 05 2015.

\bibitem[GK60]{MR115104}
James~G. Glimm and Richard~V. Kadison.
\newblock Unitary operators in {$C\sp{\ast} $}-algebras.
\newblock {\em Pacific J. Math.}, 10:547--556, 1960.

\bibitem[Gli60]{MR0112057}
James~G. Glimm.
\newblock On a certain class of operator algebras.
\newblock {\em Trans. Amer. Math. Soc.}, 95:318--340, 1960.

\bibitem[GLR85]{MR808930}
P.~Ghez, R.~Lima, and J.~E. Roberts.
\newblock {$W\sp \ast$}-categories.
\newblock {\em Pacific J. Math.}, 120(1):79--109, 1985.

\bibitem[Haa96]{MR1405610}
Rudolf Haag.
\newblock {\em Local quantum physics}.
\newblock Texts and Monographs in Physics. Springer-Verlag, Berlin, second
  edition, 1996.
\newblock Fields, particles, algebras.

\bibitem[Has04]{PhysRevB.69.104431}
M.~B. Hastings.
\newblock Lieb-schultz-mattis in higher dimensions.
\newblock {\em Phys. Rev. B}, 69:104431, Mar 2004.

\bibitem[HK64]{MR165864}
Rudolf Haag and Daniel Kastler.
\newblock An algebraic approach to quantum field theory.
\newblock {\em J. Mathematical Phys.}, 5:848--861, 1964.

\bibitem[HK06]{Hastings:2005pr}
Matthew~B. Hastings and Tohru Koma.
\newblock {Spectral gap and exponential decay of correlations}.
\newblock {\em Commun. Math. Phys.}, 265:781--804, 2006.

\bibitem[I{\etalchar{+}}24]{NonAbTop}
Mohsin Iqbal et~al.
\newblock Non-abelian topological order and anyons on a trapped-ion processor.
\newblock {\em Nature}, 626(7999):505--511, 2024.

\bibitem[JF22]{MR4444089}
Theo Johnson-Freyd.
\newblock On the classification of topological orders.
\newblock {\em Comm. Math. Phys.}, 393(2):989--1033, 2022.

\bibitem[JL24]{JL24}
Corey Jones and Junhwi Lim.
\newblock An index for quantum cellular automata on fusion spin chains.
\newblock {\em Annales Henri Poincar{\'e}}, 25(10):4399--4422, 2024.

\bibitem[JNPW23]{jones2023localtopologicalorderboundary}
Corey Jones, Pieter Naaijkens, David Penneys, and Daniel Wallick.
\newblock Local topological order and boundary algebras, 2023.

\bibitem[Jon24]{Jones2023DHRBO}
Corey Jones.
\newblock {DHR} bimodules of quasi-local algebras and symmetric quantum
  cellular automata.
\newblock {\em Quantum Topology}, 2024.

\bibitem[KGK{\etalchar{+}}14]{PhysRevX.4.031019}
M.~Kliesch, C.~Gogolin, M.~J. Kastoryano, A.~Riera, and J.~Eisert.
\newblock Locality of temperature.
\newblock {\em Phys. Rev. X}, 4:031019, Jul 2014.

\bibitem[Kit06]{KITAEV20062}
Alexei Kitaev.
\newblock Anyons in an exactly solved model and beyond.
\newblock {\em Annals of Physics}, 321(1):2--111, 2006.
\newblock January Special Issue.

\bibitem[KLT22]{MR4496385}
Yosuke Kubota, Matthias Ludewig, and Guo~Chuan Thiang.
\newblock Delocalized spectra of {L}andau operators on helical surfaces.
\newblock {\em Comm. Math. Phys.}, 395(3):1211--1242, 2022.

\bibitem[KOS03]{MR1994863}
Akitaka Kishimoto, Narutaka Ozawa, and Sh\^oichir\^o Sakai.
\newblock Homogeneity of the pure state space of a separable {$C^*$}-algebra.
\newblock {\em Canad. Math. Bull.}, 46(3):365--372, 2003.

\bibitem[Kub17]{MR3594362}
Yosuke Kubota.
\newblock Controlled topological phases and bulk-edge correspondence.
\newblock {\em Comm. Math. Phys.}, 349(2):493--525, 2017.

\bibitem[Lan17]{MR3643288}
Klaas Landsman.
\newblock {\em Foundations of quantum theory}, volume 188 of {\em Fundamental
  Theories of Physics}.
\newblock Springer, Cham, 2017.
\newblock From classical concepts to operator algebras.

\bibitem[LR72]{MR312860}
Elliott~H. Lieb and Derek~W. Robinson.
\newblock The finite group velocity of quantum spin systems.
\newblock {\em Comm. Math. Phys.}, 28:251--257, 1972.

\bibitem[LT21]{MR4287182}
Matthias Ludewig and Guo~Chuan Thiang.
\newblock Gaplessness of {L}andau {H}amiltonians on hyperbolic half-planes via
  coarse geometry.
\newblock {\em Comm. Math. Phys.}, 386(1):87--106, 2021.

\bibitem[LT22]{MR4485912}
Matthias Ludewig and Guo~Chuan Thiang.
\newblock Large-scale geometry obstructs localization.
\newblock {\em J. Math. Phys.}, 63(9):Paper No. 091902, 8, 2022.

\bibitem[Lud23]{ludewig2023coarsegeometryapplicationssolid}
Matthias Ludewig.
\newblock Coarse geometry and its applications in solid state physics, 2023.

\bibitem[Naa11]{MR2804555}
Pieter Naaijkens.
\newblock Localized endomorphisms in {K}itaev's toric code on the plane.
\newblock {\em Rev. Math. Phys.}, 23(4):347--373, 2011.

\bibitem[Naa12]{MR2956822}
Pieter Naaijkens.
\newblock Haag duality and the distal split property for cones in the toric
  code.
\newblock {\em Lett. Math. Phys.}, 101(3):341--354, 2012.

\bibitem[NOS06]{MR2256615}
Bruno Nachtergaele, Yoshiko Ogata, and Robert Sims.
\newblock Propagation of correlations in quantum lattice systems.
\newblock {\em J. Stat. Phys.}, 124(1):1--13, 2006.

\bibitem[NS06]{MR2217299}
Bruno Nachtergaele and Robert Sims.
\newblock Lieb-{R}obinson bounds and the exponential clustering theorem.
\newblock {\em Comm. Math. Phys.}, 265(1):119--130, 2006.

\bibitem[NSW13]{MR3050163}
Bruno Nachtergaele, Volkher~B. Scholz, and Reinhard~F. Werner.
\newblock Local approximation of observables and commutator bounds.
\newblock In {\em Operator methods in mathematical physics}, volume 227 of {\em
  Oper. Theory Adv. Appl.}, pages 143--149. Birkh\"{a}user/Springer Basel AG,
  Basel, 2013.

\bibitem[NY12]{MR2986138}
Piotr~W. Nowak and Guoliang Yu.
\newblock {\em Large scale geometry}.
\newblock EMS Textbooks in Mathematics. European Mathematical Society (EMS),
  Z\"{u}rich, 2012.

\bibitem[NY23]{MR4646531}
Piotr~W. Nowak and Guoliang Yu.
\newblock {\em Large scale geometry}.
\newblock EMS Textbooks in Mathematics. EMS Press, Berlin, second edition,
  [2023] \copyright2023.

\bibitem[Oga22]{MR4362722}
Yoshiko Ogata.
\newblock A derivation of braided {$C^*$}-tensor categories from gapped ground
  states satisfying the approximate {H}aag duality.
\newblock {\em J. Math. Phys.}, 63(1):Paper No. 011902, 48, 2022.

\bibitem[Pap23]{MR4619565}
Panos Papasoglu.
\newblock Polynomial growth and asymptotic dimension.
\newblock {\em Israel J. Math.}, 255(2):985--1000, 2023.

\bibitem[Pau02]{MR1976867}
Vern Paulsen.
\newblock {\em Completely bounded maps and operator algebras}, volume~78 of
  {\em Cambridge Studies in Advanced Mathematics}.
\newblock Cambridge University Press, Cambridge, 2002.

\bibitem[Qi13]{qi2013exactholographicmappingemergent}
Xiao-Liang Qi.
\newblock Exact holographic mapping and emergent space-time geometry, 2013.

\bibitem[QR19]{Qi2019determininglocal}
Xiao-Liang Qi and Daniel Ranard.
\newblock Determining a local {H}amiltonian from a single eigenstate.
\newblock {\em {Quantum}}, 3:159, July 2019.

\bibitem[Raa17]{MR3600376}
Matti Raasakka.
\newblock Spacetime-free approach to quantum theory and effective spacetime
  structure.
\newblock {\em SIGMA Symmetry Integrability Geom. Methods Appl.}, 13:Paper No.
  006, 33, 2017.

\bibitem[RLL00]{MR1783408}
M.~R{\o}rdam, F.~Larsen, and N.~Laustsen.
\newblock {\em An introduction to {$K$}-theory for {$C^*$}-algebras}, volume~49
  of {\em London Mathematical Society Student Texts}.
\newblock Cambridge University Press, Cambridge, 2000.

\bibitem[Roe93]{MR1147350}
John Roe.
\newblock Coarse cohomology and index theory on complete {R}iemannian
  manifolds.
\newblock {\em Mem. Amer. Math. Soc.}, 104(497):x+90, 1993.

\bibitem[Roe03]{MR2007488}
John Roe.
\newblock {\em Lectures on coarse geometry}, volume~31 of {\em University
  Lecture Series}.
\newblock American Mathematical Society, Providence, RI, 2003.

\bibitem[S{\etalchar{+}}21]{doi:10.1126/science.abi8378}
K.~J. Satzinger, , et~al.
\newblock Realizing topologically ordered states on a quantum processor.
\newblock {\em Science}, 374(6572):1237--1241, 2021.

\bibitem[Sak68]{MR223906}
Sh\^oichir\^o Sakai.
\newblock Derivations of simple {$C\sp{\ast} $}-algebras.
\newblock {\em J. Functional Analysis}, 2:202--206, 1968.

\bibitem[Sak20]{MR4171428}
Hiroki Sako.
\newblock Finite-dimensional approximation properties for uniform {R}oe
  algebras.
\newblock {\em J. Lond. Math. Soc. (2)}, 102(2):623--644, 2020.

\bibitem[Sch99]{doi:10.1142/S0129055X99000362}
Dirk Schlingemann.
\newblock From euclidean field theory to quantum field theory.
\newblock {\em Reviews in Mathematical Physics}, 11(09):1151--1178, 1999.

\bibitem[Sta87]{stanley1987introduction}
H.E. Stanley.
\newblock {\em Introduction to Phase Transitions and Critical Phenomena}.
\newblock International series of monographs on physics. Oxford University
  Press, 1987.

\bibitem[SW04]{https://doi.org/10.48550/arxiv.quant-ph/0405174}
B.~Schumacher and R.~F. Werner.
\newblock Reversible quantum cellular automata, 2004.

\bibitem[TWH{\etalchar{+}}23]{tomba2023boundaryalgebraskitaevquantum}
Mario Tomba, Shuqi Wei, Brett Hungar, Daniel Wallick, Kyle Kawagoe, Chian~Yeong
  Chuah, and David Penneys.
\newblock Boundary algebras of the kitaev quantum double model, 2023.

\bibitem[WW12]{WW}
Kevin Walker and Zhenghan Wang.
\newblock (3+1)-tqfts and topological insulators.
\newblock {\em Frontiers of Physics}, 7(2):150--159, 2012.

\bibitem[WY20]{MR4411373}
Rufus Willett and Guoliang Yu.
\newblock {\em Higher index theory}, volume 189 of {\em Cambridge Studies in
  Advanced Mathematics}.
\newblock Cambridge University Press, Cambridge, 2020.

\bibitem[Xo24]{NonAbTop2}
Shibo Xu and other.
\newblock Non-abelian braiding of fibonacci anyons with a superconducting
  processor.
\newblock {\em Nature Physics}, 20(9):1469--1475, 2024.

\bibitem[ZCZW19]{MR3929747}
Bei Zeng, Xie Chen, Duan-Lu Zhou, and Xiao-Gang Wen.
\newblock {\em Quantum information meets quantum matter}.
\newblock Quantum Science and Technology. Springer, New York, 2019.
\newblock From quantum entanglement to topological phases of many-body systems,
  With a foreword by John Preskill.

\end{thebibliography}

\end{document}